\documentclass[a4paper,11pt,headings=normal,abstract=true]{scrartcl}

\usepackage{amsmath}
\usepackage{amsfonts}
\usepackage{amsthm}
\usepackage{subcaption}
\usepackage{natbib}        
\usepackage{algorithm}
\usepackage{algpseudocode}
\usepackage[hidelinks]{hyperref}
\usepackage{graphicx}
\usepackage{authblk}
\usepackage{amsmath,amsfonts,amssymb}
\usepackage{siunitx}
\usepackage{subcaption}
\usepackage{booktabs}
\usepackage{bbm}
\usepackage{algorithm}
\usepackage{algpseudocode}
\usepackage{graphicx}
\usepackage[export]{adjustbox}
\MakeRobust{\subref}

\newcommand{\R}{\mathbb{R}}                                     

\newcommand{\td}[2]{\frac{\mathrm{d}#1}{\mathrm{d}#2}}          

\newcommand{\innerprod}[2]{\left\langle #1,\, #2 \right\rangle} 

\newcommand{\diff}{\mathrm{d}}                                   
\newcommand{\one}{\mathbbm{1}}                                   
\newcommand{\bias}{\mathrm{bias}}                                

\DeclareMathOperator*{\argmin}{arg\,min}

\newcommand{\vbias}{V_{\mathrm{bias}}}

\newcommand{\bE}{\mathbb{E}}

\newcommand{\cA}{\mathcal{A}}

\newcommand{\cL}{\mathcal{L}}
\newcommand{\cP}{\mathcal{P}}

\newcommand{\BB}{\mathbf{B}}
\newcommand{\BC}{\mathbf{C}}

\newcommand{\BX}{\mathbf{X}}

\newcommand{\BZ}{\mathbf{Z}}

\newcommand{\ba}{\mathbf{a}}

\newcommand{\bu}{\mathbf{u}}
\newcommand{\bv}{\mathbf{v}}

\newcommand{\bx}{\mathbf{x}}

\newcommand{\bz}{\mathbf{z}}

\newtheorem{theorem}{Theorem}[section]

\newtheorem{proposition}[theorem]{Proposition}

\newtheorem{remark}[theorem]{Remark}

\begin{document}
\title{Optimal Bias Potentials via Ergodic Optimal Control and Generator Learning}
\author[1,4]{Lei Guo}
\author[2]{Carsten Hartmann}
\author[3]{Thomas Berger}
\author[4, 5]{Feliks Nüske}
\affil[1]{Otto-von-Guericke Universität, Institute of Mathematics, Magdeburg, Germany}
\affil[2]{Brandenburg University of Technology Cottbus-Senftenberg, Institute of Mathematics, Cottbus, Germany}
\affil[3]{Martin-Luther-Universität Halle-Wittenberg, Institute of Mathematics, Halle (Saale), Germany}
\affil[4]{Max-Planck-Institute for Dynamics of Complex Technical Systems, Magdeburg, Germany}
\affil[5]{LAAS-CNRS, Université de Toulouse, France}

\maketitle

\begin{abstract}
    We investigate the computation of optimal bias potentials for accelerating transitions between metastable states and for computation of equilibrium properties in molecular dynamics simulations. We formulate optimal biasing as an ergodic optimal control problem (OCP), which can be recast as a linear eigenvalue problem for the infinitesimal generator of the unbiased dynamics. We demonstrate that data-driven learning methods for the generator enable reliable solution of the OCP, computation of biasing potentials, extraction of equilibrium properties, and acceleration of state transitions. We also explore the relation of the control problem to coarse grained representations and learning of coarse grained dynamics.
\end{abstract}

\section{Introduction}
Molecular Dynamics (MD) simulations are a powerful computational tool to study the statistical and dynamical properties of systems at the molecular scale~\cite{frenkelUnderstandingMolecularSimulation2023}. A key limitation, however, is that MD is subject to the \emph{sampling problem}: integration time steps in standard MD implementations are limited to a few femtoseconds, while relevant dynamical processes typically occur at frequencies of milliseconds or beyond. As a result, vast amounts of simulation data need to be produced to achieve sufficient sampling of a system's configurational space; for large-scale systems, exhaustive sampling of the configurational space by standard MD simulations is impossible.

A widely-used approach to circumvent this problem is \emph{biased sampling}: we model MD simulations as a Langevin process
\begin{equation}
\label{eq:sde_langevin}
    \diff \BX_t = -\nabla V(\BX_t)\,\diff t + \sqrt{2\beta^{-1}}\,\diff \BB_t,
\end{equation}
with $V:\R^d \mapsto \R$ the potential energy, $\beta > 0$ the inverse temperature, and $\BB_t$ a $d$-dimensional Brownian motion. The invariant distribution of~\eqref{eq:sde_langevin} is the Boltzmann distribution with partition function $Z$:
\begin{equation}
\label{eq:inv_measure}
    \diff \mu(\bx) = \frac{1}{Z} \exp(-\beta V(\bx))\,\diff \bx,
\end{equation}
and we assume the dynamics~\eqref{eq:sde_langevin} are ergodic with respect to $\mu$. The sampling problem can then be traced back to trapping of the dynamics in local minima of the potential $V$ (corresponding to modes of the distribution $\mu$). To overcome trapping, one can apply a bias potential $\vbias$ and modify the dynamics as
\begin{equation}
\label{eq:sde_biased}
    \diff \BX_t = -\left[ \nabla V(\BX_t) + \nabla \vbias(\BX_t)\right]\,\diff t + \sqrt{2\beta^{-1}}\,\diff \BB_t.
\end{equation}
If the bias potential is chosen well, it can facilitate escape from local minima and drastically speed up exploration of the configurational space. Highly successful variants of this approach include umbrella sampling~\cite{torrie_nonphysical_1977} (US, pre-defined bias potential), metadynamics~\cite{laio_escaping_2002} (MetaD, bias potential defined on-the-fly), and the adaptive biasing force method~\cite{darve_calculating_2001} (ABF, bias potential defined on-the-fly). Other important techniques for accelerating state space exploration include temperature-based methods such as parallel tempering~\cite{swendsen_replica_1986}, or techniques based on escape statistics, such as parallel replica~\cite{voter_parallel_1998}.

The shape of the biased dynamics~\eqref{eq:sde_biased} suggests to look for an optimal biasing potential according to pre-defined criteria. A natural question would be if there is a bias such that the state space is explored as quickly as possible while perturbing the physical dynamics~\eqref{eq:sde_langevin} as little as possible. This leads to a re-formulation of~\eqref{eq:sde_biased} as an \emph{optimal control problem}, where the bias is interpreted as a control policy that is optimized towards a specified goal. Control-theoretic approaches to biased sampling have been explored in~\cite{schutte_optimal_2012} using Markov state models, in~\cite{hartmann_efficient_2012} using cross-entropy methods, in~\cite{vandeneijnden_rare_2012} using small-noise asymptotics, and in~\cite{holdijk_stochastic_2023} 
using Schrödinger bridges. 

\paragraph{Contributions.} Our work follows~\cite{schutte_optimal_2012}: we study biased sampling as an \emph{ergodic optimal control problem} of the form
\begin{equation}
\label{eq:ergodic_ocp}
    J^*(\bx) = \min_\bu \limsup_{T\rightarrow \infty} \bE\left[\frac{1}{T}\int_0^T \sigma \ell(\BX_s) + \frac{\beta}{4}\|\bu_s\|^2\,\diff s \,\vert \,\BX_0 = \bx  \right],
\end{equation}
where $\bx \in \R^d$ is a position, $\sigma > 0$ is a parameter, and $\bu$ is a control policy, that is a time-dependent stochastic process taking values in $\R^d$. As the guiding example for this paper, we will choose the running cost $\ell$ as the indicator function of the complement of a target set $A$,
\begin{equation*}
    \ell(\bx) = \one_{\bar{A}}(\bx)
\end{equation*}
that is, the goal is to drive the dynamics towards the target set $A$. The infinite time horizon in~\eqref{eq:ergodic_ocp}, combined with ergodicity, allows to reduce this problem to a linear eigenvalue problem for the infinitesimal generator $\cL$ of the physical dynamics~\eqref{eq:sde_langevin} (see Sec.~\ref{ssec:generator}), namely the smallest (ground-state) eigenvalue $\lambda > 0$ of
\begin{equation}
\label{eq:ev_problem_intro}
    (\sigma \ell - \cL)\varphi = \lambda \varphi\,,
\end{equation}
where the eigenvalue $\lambda=\lambda(\sigma)$ satisfies (cf.~\cite[Sec.~5.1]{schutte_optimal_2012})
\begin{equation}
\label{eq:limit_lambda_sigma}
    \left.\frac{\diff \lambda}{\diff \sigma}\right|_{\sigma=0} = \lim_{T\to\infty}\frac{1}{T}\int_0^T \ell(\bx_s)\,\diff s\,.
\end{equation}
Equation~\eqref{eq:limit_lambda_sigma} allows to recover equilibrium expectations of the cost function $\ell$ from the eigenvalues $\lambda$.

The corresponding eigenfunction $\varphi=\varphi(\cdot;\sigma)$ directly leads to the optimal biasing potential $\vbias$. In this paper, we make the following contributions:
\begin{enumerate}
    \item We show that the linear eigenvalue problem~\eqref{eq:ev_problem_intro} can be solved efficiently using modern machine learning techniques to approximate the generator $\cL$ of the physical dynamics, including generator extended dynamic mode decomposition~\cite{klus_data-driven_2020} (gEDMD) and neural network approximations. These learning algorithms require equilibrium samples from the Boltzmann distribution for training.
    \item As generator learning algorithms do not require time series data, samples from established enhanced sampling algorithms can be used for training, including simulations at higher temperatures or metadynamics simulations.
    \item We demonstrate that solutions obtained via generator learning are robust to training data size and parameter variations. We analyze their performance using multiple model problems in low dimension, as well as molecular dynamics simulations of the alanine dipeptide. We show that equilibrium expectations can be recovered robustly by extrapolating~\eqref{eq:limit_lambda_sigma}. Moreover, the resulting bias potentials lead to a drastic reduction of the mean first passage time into the target set $A$.
    \item We explain how the resulting method can be understood in the context of coarse-grained (CG) dynamics, namely as applying an optimal forcing in CG space.
\end{enumerate}

\paragraph{Open Problems.} The suggested approach provides a rich data-driven framework based on a rigorous mathematical footing provided by optimal control theory. However, it is not a complete solution to the sampling problem, as training data from the equilibrium distribution are still required. Strategies to obtain these data in an adaptive manner will be explored in future research.

\section{Dynamic Optimal Control Problems and the Generator}\label{sec:theory}

\subsection{Langevin Dynamics and Control Problems}
We study the Langevin equation~\eqref{eq:sde_langevin}. Most of the following can be extended to more general reversible processes, but we restrict the presentation to Langevin dynamics for simplicity.

We seek a minimally invasive \emph{forcing} or \emph{bias} that steers the dynamics towards optimizing a prescribed goal. Concretely, we consider the biased dynamics 
\begin{equation}
\label{eq:sde_controlled}
    \diff \BX_t = \left[-\nabla V(\BX_t) + \bu_t\right]\,\diff t + \sqrt{2\beta^{-1}}\,\diff \BB_t,
\end{equation}
where $\bu_t$ is a (a priori time-dependent) control input taking values in $\R^d$. We introduce a running cost (or stage cost) \begin{equation}\label{eq:running_cost}
    c: \R^d \times \R^d \to \R_{\geq 0},\ (\bx,\bu)\mapsto \sigma \ell(\bx) + \frac{\beta}{4}\|\bu\|^2,
\end{equation}
where $\sigma > 0$ is a given parameter (that will be varied later on) and $\ell: \R^d \mapsto \R_{\geq 0}$ is a non-negative state cost. The quadratic term penalizes excessive forcing to keep the controlled dynamics as close as possible to the physical equation, $\sigma$ sets the relative weight of the state cost against this control-effort penalty, while the specific coefficient $\beta/4$ is fixed by the linearization argument given below. The objective then is to minimize the long-time averaged cost subject to the biased dynamics~\eqref{eq:sde_controlled}:
\begin{equation}
\label{eq:ergodic_ocp}
\begin{split}
    J^*(\bx) &= \min_\bu \limsup_{T\rightarrow \infty} \bE\left[\frac{1}{T}\int_0^T c(\BX_s, \bu_s)\,\diff s \,\vert\, \BX_0 = \bx \right] = \min_\bu J(\bx, \bu), \\
    \text{s.t.} & \quad \BX_s \text{ solution of~\eqref{eq:sde_controlled}}.
\end{split}
\end{equation}

\paragraph{Guiding Example.} Throughout the text, we consider the problem of forcing the physical dynamics~\eqref{eq:sde_langevin} to spend as much time as possible in a fixed subset $A \subset \R^d$ of its state space. To this end, we choose the cost function
\begin{equation}
    \label{eq:guiding_example_cost}
    \one_{\bar{A}}(\bx) + \frac{\beta}{4} \|\bu\|^2,
\end{equation}
where $\one_A$ is the indicator function of the set $A$ and $\bar{A}$ is its complement. The set $A$ will typically be a state of physical significance, such as a metastable state of the physical dynamics~\eqref{eq:sde_langevin}. Figure~\ref{fig:example_dw} illustrates the problem on a one-dimensional double-well potential $V$: under the uncontrolled dynamics the stationary distribution $\mu(\bx) \propto e^{-\beta V(\bx)}$ places equal mass in both wells; the optimal controller reshapes the potential $V$ to $V_{\mathrm{opt}}(\bx)$, concentrating the invariant measure on $A$ while minimizing the control cost.

\begin{figure}
    \centering
\includegraphics[width=\linewidth]{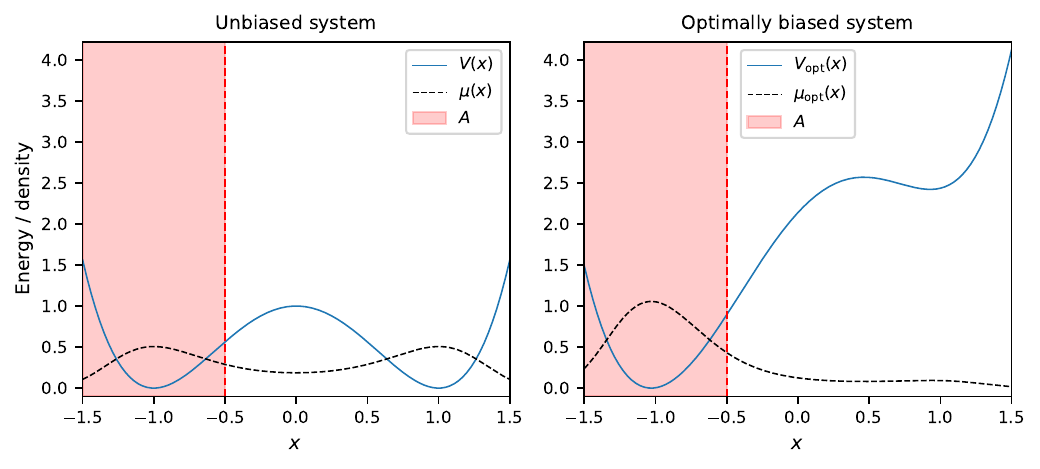}
    \caption{\textbf{Illustration of the ergodic control problem} on the one-dimensional double-well potential at inverse temperature $\beta=1$. Left: uncontrolled system with potential $V(x)$ and stationary distribution $\mu(\bx) \propto e^{-\beta V(\bx)}$. The target set $A$ (shaded) is the left well region; under the uncontrolled dynamics, the system spends roughly equal time in both wells. Right: optimally controlled system with reshaped potential $V_{\mathrm{opt}}(\bx)$ and the corresponding stationary distribution $\mu_{\mathrm{opt}}(\bx) \propto e^{-\beta V_{\mathrm{opt}}(\bx)}$, concentrating the invariant measure on $A$.}
    \label{fig:example_dw}
\end{figure}

\subsection{The Generator and the Optimal Forcing}
\label{ssec:generator}
The minimization problem~\eqref{eq:ergodic_ocp} is a stochastic \emph{optimal control problem (OCP)}. A central tool we will need to describe its solution is the \emph{Koopman generator} $\cL$ associated to the dynamics~\eqref{eq:sde_langevin}. The differential operator $\cL$ is defined by
\begin{equation}
    \label{eq:generator}
    \cL \phi = -\nabla V \cdot \nabla \phi + \frac{1}{\beta}\Delta \phi,
\end{equation}
acting on smooth observable functions $\phi: \R^d \to \R$. The generator is symmetric with respect to the Boltzmann-weighted inner product
\begin{equation*}
    \innerprod{\phi}{\psi}_\mu = \int_{\R^d} \phi(\bx)\psi(\bx)\,\diff \mu(\bx),
\end{equation*}
and we have the integration-by-parts formula~\cite{lelievre_partial_2016}:
\begin{equation}
\label{eq:integration_by_parts}
    -\innerprod{\cL \phi}{\psi}_\mu = -\innerprod{\phi}{\cL\psi}_\mu = \frac{1}{\beta}\int_{\R^d} \nabla \phi(\bx) \cdot \nabla \psi(\bx)\,\diff \mu(\bx).
\end{equation}
On a properly chosen domain, the (negative) generator $-\cL$ is non-negative and self-adjoint. Its spectrum is real-valued, non-negative and discrete~\cite{bakry_analysis_2014,lelievre_partial_2016}.

Optimal control theory characterizes the optimal cost-to-go $J^{*}$ from~\eqref{eq:ergodic_ocp}, as well as the associated optimal forcing $\bu^* = \argmin_\bu J(\bx, \bu)$, via the following result. It follows from the theory of ergodic control of non-degenerate diffusions under a near-monotone running cost, see \cite{arapostathis2012ergodic} for a detailed mathematical treatment; cf.~\cite{schutte_optimal_2012}.

\begin{proposition}
\label{prop:ev_problem_gen}
    The optimal cost-to-go $J^*(\bx)$ in~\eqref{eq:ergodic_ocp} satisfies the following properties:    \begin{itemize}
        \item[($i$)] The cost $J^*(\bx)$ is \textbf{independent} of $\bx$, i.e. $J^*(\bx) = \lambda \geq 0$ for all $\bx \in \R^d$.
        \item[($ii$)] The optimal control $\bu^*$ is a \textbf{stationary feedback control}, expressible as the gradient of a biasing potential $V_{\bias}$:
        \begin{equation*}
            \bu^*_s = -\nabla V_\bias(\BX_s).
        \end{equation*}
        \item[($iii$)] The optimal cost-to-go $\lambda$ and the biasing potential $V_\bias$ can be computed from the ground state of the following \textbf{linear eigenvalue equation}:
    \begin{equation}
        \label{eq:ev_problem}
        (\sigma \ell - \cL)\varphi = \lambda \varphi
    \end{equation}
    for the uncontrolled generator $\cL$ defined in~\eqref{eq:generator}. The smallest eigenvalue $\lambda$ equals the optimal cost $J^{*}(\bx)$ for any $\bx$. By the Perron-Frobenius Theorem, the eigenfunction $\varphi$ is strictly positive, and gives rise to the bias potential as
    \begin{equation*}
        V_\bias(\bx) = - \frac{2}{\beta}\log \varphi(\bx).
    \end{equation*}
   \end{itemize}
    The optimal controlled dynamics is then obtained by choosing
    \begin{equation*}
    \begin{split}
        \diff \BX^*_t
        &= \left[-\nabla V(\BX^*_t) + \frac{2}{\beta}\frac{\nabla \varphi(\BX^*_t)}{\varphi(\BX^*_t)}\right]\,\diff t + \sqrt{2\beta^{-1}}\,\diff \BB_t.
    \end{split}
    \end{equation*}
\end{proposition}

The essential content of Proposition~\ref{prop:ev_problem_gen} is that the optimal forcing is fully characterized by a \emph{linear equation} involving the \emph{physical} generator $\cL$ — an intractable nonlinear control problem is reduced to linear algebra. Intuitively, the ergodic (long-time) limit in~\eqref{eq:ergodic_ocp} requires only a stationary control, and also removes the dependence on the initial condition; deriving~\eqref{eq:ev_problem} then amounts to writing out the Hamilton–Jacobi–Bellman (HJB) equation for~\eqref{eq:ergodic_ocp} and applying a logarithmic transformation that linearizes it. We give a full derivation for the reader's convenience in Appendix A. Sec.~\ref{sec:num} shows how to exploit~\eqref{eq:ev_problem} numerically based only on equilibrium samples of the physical dynamics.

\paragraph{The prefactor $\beta/4$.} Replacing the specific prefactor $\beta/4$ in~\eqref{eq:running_cost} by a general coefficient $\eta$ leads to re-scaling of the generator in~\eqref{eq:ev_problem}, which amounts to a re-scaling of time in~\eqref{eq:sde_langevin}. We stick to the case $\frac{\beta}{4}$ for simplicity.

\subsection{Evaluation of Equilibrium Averages}\label{sec:equilibrium-average}
The solution of the linear eigenvalue problem~\eqref{eq:ev_problem} can subsequently be used to compute equilibrium averages of the unbiased dynamics~\eqref{eq:sde_langevin}. To this end, we highlight the dependence of the eigenvalue $\lambda$ on the pre-factor $\sigma$ by writing $\lambda(\sigma)$. By a logarithmic transformation (e.g.~\cite{schutte_optimal_2012}), the eigenvalue $\lambda(\sigma)$ can equivalently be written as 
\begin{equation}\label{eq:cgf}
\lambda(\sigma) = -\limsup_{T\to\infty}\frac{1}{T} \log v_\sigma(\bx,T),\quad \sigma > 0\,,
\end{equation}
where $v_\sigma$ denotes the following exponential average:
\begin{equation}\label{eq:mgf}
v_\sigma(\bx, T) = \bE\left[\exp\left(-\sigma \int_{0}^{T}\ell(\BX_{s})\,\diff s\right)\bigg|\,\BX_{0}=\bx\right],
\end{equation}
and $\BX_s$ is the solution of the physical Langevin dynamics~\eqref{eq:sde_langevin}. The map $\sigma\to v_\sigma(\bx, T)$, for fixed $\bx$ and $T$, is the moment generating function of the random variable 
\begin{equation*}
L_{T} = -\int_{0}^{T}\ell(\BX_{s})\, \diff s\,,
\end{equation*} 
parametrized by the initial datum $\bx$ and the terminal time $T$. Assuming that the moment generating function $v_\sigma$ is finite and strictly positive for all $\sigma$ close enough to zero, we can expand its logarithm about $\sigma = 0$: 
\begin{align*}
    \log v_\sigma(\bx, T) = - \sigma \bE[L_T\,\vert\,\BX_0 = \bx] + \frac{\sigma^2}{2}\left(\bE[L_T^2\,\vert\,\BX_0 = \bx] - \left(\bE[L_T\,\vert\,\BX_0 = \bx]\right)^2\right) - \ldots\,.
\end{align*}
Taking the limit $T \to \infty$, it follows that 
\begin{equation*}
\lambda(\sigma) = \limsup_{T\to\infty}\frac{\sigma}{T} \bE\left[L_{T}\,\vert\,\BX_0 = \bx \right] - \limsup_{T\to\infty}\frac{\sigma^{2}}{2 T} \left(\bE[L_T^2\,\vert\,\BX_0 = \bx] - \left(\bE[L_T\,\vert\,\BX_0 = \bx]\right)^2\right) + \ldots\,,
\end{equation*}
where all remaining terms are higher order in $\sigma$. Then, taking the derivative with respect to $\sigma$ at $\sigma = 0$, and using ergodicity of the Langevin equation~\eqref{eq:sde_langevin}, we find
\begin{align*}
    \left.\td{}{\sigma} \lambda(\sigma)\right|_{\sigma = 0} &= -\limsup_{T\to\infty}\frac{1}{T} \bE\left[L_{T} \,\vert\,\BX_0 = \bx\right] \\
    &= \int_{\R^d} \ell(\bx)\,\diff\mu(\bx) = \mathbb{E}^\mu[\ell(\bx)].
\end{align*}

\begin{remark}
In our experiments, we will estimate the equilibrium average $\mathbb{E}^\mu[\ell(\bx)]$ by extrapolating $\left.\td{}{\sigma} \lambda(\sigma)\right|_{\sigma = 0}$ from solutions of the eigenvalue problem~\eqref{eq:ev_problem} at different $\sigma$ close to zero. Alternatively, one can also run controlled simulations at small $\sigma$, use them to evaluate the cost $J^* = \lambda(\sigma)$ defined in~\eqref{eq:ergodic_ocp}, and extrapolate the rate to $\sigma = 0$. However, since $\varphi$ becomes constant in the limit $\sigma\to 0$, resulting in vanishing of the control force, there will be a trade-off between the simulation speed-up from using controlled simulations and a low-variance evaluation of the quantity of interest.
\end{remark}

\section{Numerical Solution}\label{sec:num}
We now consider the numerical approximation of the eigenvalue problem~\eqref{eq:ev_problem}. The key ingredient is the Rayleigh quotient of the symmetric operator 
$$
\mathcal{A}:=\sigma\ell-\cL,
$$
which is defined for any nonzero trial function $\phi$ by:
\begin{equation}\label{eq:rayleigh}
R[\phi] := \frac{\langle\phi,\mathcal{A}\phi\rangle_\mu}{\langle\phi,\phi\rangle_\mu}
= \frac{\sigma\int\ell\,\phi^2 \,\diff\mu + \tfrac1\beta\int\nabla\phi\nabla\phi^{\top}\diff\mu}{\int\phi^2 \,\diff\mu}.
\end{equation}
The smallest eigenvalue $\lambda$ in~\eqref{eq:ev_problem} is the minimizer of the Rayleigh quotient, see e.g.~\cite{teschl_mathematical_2014}:
\begin{equation}
    \label{eq:variational_rayleigh}
    \lambda = \min_{\phi \ne 0}R[\phi].
\end{equation}

This variational formulation provides a common starting point for the discretizations we consider below.

\subsection{Generator Extended Dynamic Mode Decomposition}\label{sec:gedmd}
Generator extended dynamic mode decomposition (\emph{gEDMD})~\cite{klus_data-driven_2020} provides an empirical data-driven approximation of the Koopman generator based on equilibrium samples. Here we will use gEDMD to compute the principal eigenvalue of $\mathcal{A}$ from the weak form of the eigenvalue problem (\ref{eq:ev_problem_intro}):  
given a set of $n$ basis functions $\psi_i: \R^d \to \R$, 
\begin{equation*}
    \psi(\bx) = \begin{bmatrix}
        \psi_1(\bx) & \ldots & \psi_n(\bx)
    \end{bmatrix} \in \R^n,
\end{equation*}
spanning an $n$-dimensional approximation subspace, and $m$ samples drawn from the Boltzmann distribution $\mu$, 
\begin{equation*}
    \BX = \begin{bmatrix}
        \bx_1 & \ldots & \bx_m
    \end{bmatrix} \in \R^{d \times m},
\end{equation*}
we form the following empirical \emph{mass and stiffness matrices} of dimension $n \times n$:
\begin{align}
\label{eq:estimators_generator}
    \BC^m_0 &= \frac{1}{m}\sum_{k=1}^m \psi(\bx_k) \otimes \psi(\bx_k), & \BC^m_\cL &= -\frac{1}{\beta m}\sum_{k=1}^m \nabla \psi(\bx_k) \cdot \nabla \psi(\bx_k)^\top,
\end{align}
where $\otimes$ denotes the outer product between vectors. Alongside $\BC_{0}^{m}$ and $\BC_{\cL}^{m}$, we form an additional matrix accounting for the zero-order term involving the running cost:
\begin{equation}
\label{eq:gedmd_matrix_cost}
    \BC^m_\ell = \sum_{k=1}^m \ell(\bx_k)\psi(\bx_k) \otimes \psi(\bx_k).
\end{equation}
Finally, we solve the generalized matrix eigenvalue problem 
\begin{equation}
\label{eq:gev_generator}
    (\sigma \BC^m_\ell - \BC^m_\cL) \bv^m = \hat{\lambda}^m \BC^m_0 \bv^m,
\end{equation}
and extract its smallest eigenvalue $\hat{\lambda}^m$, which serves as approximation to the true optimal cost-to-go $\lambda$. Assuming that the chosen basis is such that the eigenvector that corresponds to the eigenvalue $\hat{\lambda}^m$ is strictly positive,   the associated eigenvector $\bv^m$ provides an approximation to the bias potential by 
\begin{equation*}
    V^m_{\mathrm{bias}}(\bx) = - \frac{2}{\beta}\log\left[\sum_{i=1}^n \bv^m_i \psi_i(\bx) \right].
\end{equation*}
In the limit of infinite training data ($m\rightarrow \infty$), the eigenvalue $\hat{\lambda}^m$ converges to a limiting eigenvalue $\hat{\lambda}$~\cite{nuske_finite-data_2023}. The associated eigenvector $\bv$ encodes the minimizer of the Rayleigh quotient~\eqref{eq:rayleigh} restricted to the $n$-dimensional trial space. By the variational characterization~\eqref{eq:variational_rayleigh}, we have
 \begin{equation*}
        \lambda \leq \hat{\lambda}.
    \end{equation*}

   


\subsection{Neural Cole--Hopf Representation}\label{sec:nn}
Instead of inferring $\vbias$ as a linear combination of fixed basis functions, we also consider a nonlinear parametrization. Furthermore, as the true eigenfunction $\varphi$ in~\eqref{eq:ev_problem} is positive, we consider a positive parametrization based on a Cole--Hopf transformation,
\begin{equation}
\label{eq:cole-hopf-nn}
\varphi_\theta(\bx) := e^{-W_\theta(\bx)},
\end{equation}
where $W_\theta:\R^d\to\R$ is represented by a neural network.

For samples $\bx_1,\dots,\bx_m$ drawn from the stationary distribution $\mu$, the network parameters are obtained by minimizing the empirical Rayleigh quotient
\begin{equation}
\label{eq:rayleigh-nn}
R_m(\theta) = \frac{\sum_k\sigma\ell(\bx_k)\varphi_\theta(\bx_k)^2
+ \tfrac1{\beta}\sum_k\|\nabla\varphi_\theta(\bx_k)\|^2}
{\sum_k\varphi_\theta(\bx_k)^2}.
\end{equation}
Spatial derivatives $\nabla\varphi_\theta$ and derivatives of the objective with respect to the neural network parameters are
evaluated using automatic differentiation.

The Rayleigh quotient~\eqref{eq:rayleigh-nn} is invariant under multiplication of the eigenfunction by a nonzero constant $c$, $R_{m}[c\varphi_\theta]=R_{m}[\varphi_\theta]$, since both
its numerator and denominator scale by $c^2$. Hence, the normalization of the eigenfunction is not determined by Rayleigh-quotient minimization. We remove this redundancy by fixing a reference point $\bx_0$ and by enforcing $W_\theta(\bx_0)=0$, and therefore $\varphi(\bx_0)=1$, exactly for every value of $\theta$. This gauge fixes only the arbitrary normalization of the eigenfunction and does not alter the Rayleigh quotient or the associated control force, which depends on spatial derivatives of $W_\theta$.

The resulting finite-dimensional optimization problem 
$$
\theta^{*}\in\argmin_{\theta} R_m(\theta)
$$
is solved using the Adam optimizer. Since the optimization problem is non-convex, the calculations are repeated for multiple independent random initializations.

The neural and gEDMD approaches target the same variational problem. Their difference is in the trial space: gEDMD uses a fixed linear space with a closed-form generalized eigenvalue problem, whereas the neural approach uses a nonlinear, non-convex parametrization with positivity built in.



\subsection{Re-weighting}\label{sec:reweight}
Estimating either the gEDMD matrices~\eqref{eq:estimators_generator}-\eqref{eq:gedmd_matrix_cost} or the empirical Rayleigh quotient~\eqref{eq:rayleigh-nn} requires samples from the equilibrium distribution at the target inverse temperature $\beta$. For systems with large energetic barriers, direct sampling at the target temperature can be prohibitively expensive. We therefore use importance re-weighting to construct target-temperature averages from samples generated under an alternative equilibrium distribution. For training data sampled from an alternative distribution $\nu$, we denote the likelihood ratio (or re-weighting factor) between the Boltzmann distribution $\mu$ in~\eqref{eq:inv_measure} and $\nu$ by
\begin{equation*}
    w(\bx) = \frac{\diff\mu}{\diff\nu}(\bx).
\end{equation*}
Then every sample $\bx_k$ in the gEDMD matrices~\eqref{eq:estimators_generator}-\eqref{eq:gedmd_matrix_cost} or in the empirical Rayleigh quotient~\eqref{eq:rayleigh-nn} is weighted by $w(\bx_k)$.

\subsubsection{Temperature Re-weighting}\label{sec:temp_reweight}
We write $\mu_\beta$ for the equilibrium distribution~\eqref{eq:inv_measure} associated to the inverse temperature $\beta$. Let $\beta_0<\beta$ denote a reference inverse temperature. Samples generated from $\mu_{\beta_0}$ can be reweighted to the target distribution $\mu_{\beta}$ using the likelihood ratio
\begin{equation}
\label{eq:temp_reweight}
w(\bx) = \exp\left(-(\beta-\beta_0)V(\bx)\right)\propto\frac{\diff\mu_{\beta}}{\diff\mu_{\beta_0}}(\bx).
\end{equation}

\subsubsection{Bias-Potential Re-weighting}\label{sec:bias_reweight}
For molecular systems, direct temperature re-weighting typically suffers from insurmountable variance. Most popular enhanced sampling methods lead to a biased dynamics
$$
    \diff \BX_t = -\nabla( V(\BX_t)+F_{\mathrm{bias}}(\BX_t))\,\diff t + \sqrt{2\beta^{-1}}\,\diff \BB_t
$$
with stationary density $\mu_{\mathrm{bias}}\propto e^{-\beta(V+F_{\mathrm{bias}})}$. The corresponding reweighting factor is then
\begin{equation}
\label{eq:reweighting-bias}
w(\bx) = \exp\big(\beta F_{\mathrm{bias}}(\bx)\big).
\end{equation}
In metadynamics, the bias potential is the cumulative sum of all deposited Gaussian hill potentials. The re-weighting factors $w(\bx_k)$ can be directly extracted from the simulation data for all samples $\bx_k$.

Averages of observables $\phi$ with respect to the equilibrium distribution $\mu$ can then be recast as the self-normalised importance sampling estimator
\begin{equation}\label{eq:SNIS}
    \mathbb{E}^\mu[\phi(\bx)] = \frac{\mathbb{E}[\phi(\bx)w(\bx)] }{\mathbb{E}[w(\bx)]}  \approx \frac{\sum_k\phi(\bx_k)w(\bx_k) }{\sum_k w(\bx_k)} \,.
\end{equation}

\section{Low-dimensional Model Systems}\label{sec:num_res}
We begin by evaluating the method on three low-dimensional benchmark potentials of increasing complexity: a one-dimensional double-well, a two-dimensional lemon slice potential, and a two-dimensional three-hole potential. We introduce these systems below.


\subsection{Systems}\label{sec:sys}
\paragraph{One-dimensional Double-Well.} We consider
$$
 V(x):=(x^{2}-1)^{2},\quad x\in\R,
$$
with metastable wells located at $x=\pm1$ and a saddle point at $x=0$. The target region is $A:=(-\infty, -0.5]$, so the controller is expected to concentrate the invariant measure near the left well.

\paragraph{Two-dimensional Lemon Slice.} We write $(x,y)\in\R^2$ in polar coordinates
$$
r=\sqrt{x^{2}+y^{2}},\quad \phi =\arctan2(y, x),
$$
and consider
$$
V(x, y):=\cos(k\phi) + \frac{1}{\cos(0.5\phi)} + 10(r-1)^{2} + \frac{1}{r},\quad k=4.
$$
The potential contains two deep minima near $(0.7, \pm 0.7)$ and two shallow minima near $(-0.6, \pm 0.8)$. The target set is the neighborhood of the upper-left shallow minimum, 
$$
A=\Big\{(x,y)\in\R^2\,\Big|\,\Big|\arctan2(y,x)-\frac{3\pi}{4}\Big|\le\frac{\pi}{6},\,\,V(x,y)\le4.5\Big\}.
$$

\paragraph{Two-dimensional Three-Hole Potential.} For the third benchmark we consider
$$
\begin{aligned}
 V(x,y)&=3e^{-x^{2}-(y-\frac{1}{3})^2}-3e^{-x^{2}-(y-\frac{5}{3})^{2}}-5e^{-(x-1)^{2}-y^{2}}\\
 &\quad-5e^{-(x+1)^{2}-y^{2}}+ 0.2x^{4}+0.2(y-\frac{1}{3})^{4},\qquad\quad (x,y)\in\R^2.
\end{aligned}
$$
This potential contains two deep minima approximately at $(\pm1,0)$ and a shallow minimum approximately at $(0, 1.5)$. The target set is the upper shallow well
$$
A=\{(x,y)\in\R^2\,\big|\,y>1.0,\,\,V(x,y)\le-1.8\}.
$$

\begin{figure}[h!]
    \centering
\includegraphics[width=\linewidth]{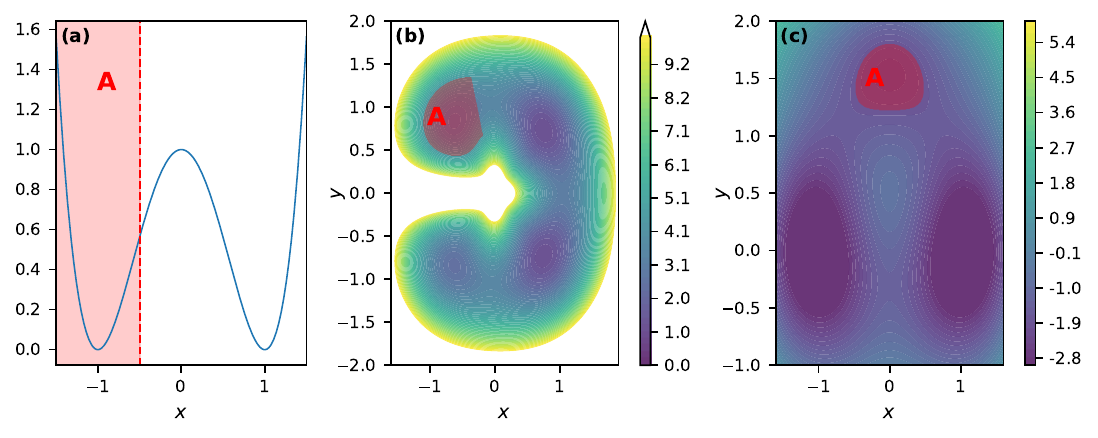}
    \caption{\textbf{Three benchmark systems}: (a) double-well; (b) lemon slice; (c) three-hole. The shaded area indicates the target set $A$ in each case.}
    \label{fig:sys}
\end{figure}

\subsection{Computational Settings}

\subsubsection{Data Generation}
We consider inverse temperatures $\beta\in[0.5, 4.0]$ and running-cost prefactors $\sigma\in[0.0, 2.0]$. When training directly at the target temperature, training data are generated by rejection sampling from the Boltzmann distribution restricted to a finite computational domain. The sampling domains are chosen to contain the metastable wells and the transition regions relevant to the target sets; the specific domain used for the double-well system is $[-2.0,2.0]$, for the lemon slice system it is $[-1.5, 1.5]\times[-1.5,1.5]$ and for the three-hole system it is $[-1.6, 1.6]\times[-1.0,2.0]$. For temperature re-weighting, the physical dynamics are instead sampled at a low reference inverse temperature $\beta_0$ and the resulting data are re-weighted according to~\eqref{eq:temp_reweight}.

\subsubsection{Basis Sets}\label{sec:rff}
A first choice for the basis set in~\eqref{eq:gev_generator} is as a collection of \emph{random Fourier features (RFF)}~\cite{rahimiRandomFeaturesLargescale2007a}: after sampling $n$ random frequencies $\omega_i \in \R^d$ according to a fixed spectral distribution $\rho$ on $\R^d$, and additionally drawing $n$ random phase shifts $b_i \in [0, 2\pi]$ uniformly, we form the basis functions as
\begin{equation}
\label{eq:random_features}
    \psi_i(\bx) = \cos\left(\omega_i^\top \bx + b_i \right),\quad i=1,\ldots,n\,.
\end{equation}
In the limit of large $n$, this approximation converges to a kernel discretization of the operator $\cA$, see~\cite{nuske_efficient_2023} for details. In particular, if the spectral distribution is normal, the large $n$ limit is a Gaussian kernel discretization.

As a deterministic alternative, we use tensor-product cubic B-splines of the form
\begin{equation*}
    \psi_{i_1\cdots i_d}(\bx)=\prod_{j=1}^{d}B_{i_j}^{(j)}(\bx^j),\quad n=\prod_{j=1}^{d} N_j,
\end{equation*}
where $B_{ij}^{(j)}(\bx^j)$ are one-dimensional B-splines on intervals $[a_j, b_j]$ along each coordinate direction $j=1,\dots,d$. The compact support of the B-splines yields a locally supported representation and correspondingly sparse matrix structure. As shown in Sec.~\ref{sec:opt-cost}, the RFF and B-spline bases give comparable estimates of the eigenvalue $\lambda$, whereas the B-spline basis provides a more robust reconstruction of the eigenfunction and the resulting bias potential. However, the tensor product construction is only feasible for low-dimensional examples.

For the numerical experiments, the RFF basis~\eqref{eq:random_features} uses a Gaussian spectral distribution with the number of features and bandwidth summarized in Table~\ref{tab:rff}. The corresponding truncated domain, spline degree, and basis size used for B-spline calculations are listed in Table~\ref{tab:bsp}. 

\begin{table}[t]
\centering
\begin{tabular}{lcc}
\toprule
System & $n$ & $\gamma_{\mathrm{rff}}$ \\
\midrule
Double-well & 50  & 0.3 \\
Lemon slice & 100 & 0.3 \\
Three-hole  & 100 & 0.35 \\
\bottomrule
\end{tabular}
\caption{Parameters used for the random Fourier feature basis. $n$ is the number of random features and $\gamma_{\mathrm{rff}}$ is the bandwidth of the Gaussian spectral distribution.}
\label{tab:rff}
\end{table}

\begin{table}[t]
\centering
\begin{tabular}{lcccc}
\toprule
System &Domain $\Omega$ & Degree & $N_j$ per dimension & Total $n=\prod_jN_j$ \\
\midrule
Double-well &$[-2, 2]$  & cubic & -- & $30$ \\
Lemon slice &$[-1.5, 1.5]\times[-1.5,1.5]$ & cubic & $15\times15$ & $225$ \\
Three-hole &$[-1.6, 1.6]\times[-1.0,2.0]$ & cubic & $10\times10$ & $100$ \\
\bottomrule
\end{tabular}
\caption{B-spline basis: spline
degree, and number of univariate splines per coordinate direction, for each
benchmark system.}
\label{tab:bsp}
\end{table}

\subsubsection{Neural Network Parameters}
For the neural Cole--Hopf (NCH) approximation (Sec.~\ref{sec:nn}), $W_\theta$ is represented
by a single-hidden-layer neural network with hyperbolic-tangent activation
functions. The network parameters are optimized by minimizing the empirical Rayleigh quotient~\eqref{eq:rayleigh-nn} using Adam~\cite{kingma_adam_2014}, with spatial
derivatives and parameter gradients evaluated by automatic differentiation
using \textsc{JAX}~\cite{bradbury_jax_2018}. The training data are generated by rejection sampling from the Boltzmann distribution at the target inverse temperature $\beta$. The numerical experiments are repeated for 5 independent network initializations, and the reported NCH results are averaged over these runs. Since the lemon slice and three-hole systems use different network widths and optimization schedules, the corresponding parameters are summarized in Table~\ref{tab:nch_parameters}.

\begin{table}[t]
\centering

\begin{tabular}{lcc}
\toprule
Parameter & Lemon slice & Three-hole \\
\midrule
$\beta$ & $3.0$ & $2.5$ \\
$\sigma$ & $0.7$ & $2.0$ \\
Training samples $m$ & $2\times10^{4}$ & $2\times10^{4}$ \\
Hidden units & $128$ & $64$ \\
Training epochs & $4000$ & $4000$ \\
Initial learning rate & $2\times10^{-3}$ & $2\times10^{-3}$ \\
Learning-rate decay factor & $0.5$ & $0.5$ \\
Decay epochs & $500,\,1000,\,2000$ & $2500,\,3000$ \\
Independent seeds & $5$ & $5$ \\
\bottomrule
\end{tabular}
\caption{
Numerical parameters used for the neural Cole--Hopf (NCH) calculations.
Training data are generated by rejection sampling from the Boltzmann
distribution at the target inverse temperature $\beta$. For each system, the
network is trained independently from 5 random initializations, and the
reported NCH results are averaged over the 5 runs.
}
\label{tab:nch_parameters}
\end{table}

\subsubsection{Reference Solution via Finite Elements}\label{sec:fem}
As a reference solution, we additionally solve~\eqref{eq:ev_problem} directly by a finite element (FEM) discretization of its weak form, implemented in \textsc{FEniCS} / \textsc{Dolfinx}~\cite{baratta_dolfinx_2023}, with meshes generated by Gmsh~\cite{geuzaine_gmsh_2009}.  No essential boundary condition is imposed. The weak formulation therefore
corresponds to the natural weighted Neumann condition $e^{-\beta V(x)}\nabla\varphi(x)\cdot n(x)=0$
on the boundary, equivalently a no-flux condition for the reversible dynamics. The computational domains are chosen sufficiently large that boundary effects are negligible in the region of interest. Mesh-refinement calculations are performed before selecting the production meshes listed in Table~\ref{tab:fem-domains}.

\begin{table}[t]
\centering
\begin{tabular}{lccc}
\toprule
System & Domain $\Omega$ & Element &  Mesh Size $h$ \\
\midrule
Double-well & $[-2,2]$ & P1 (1D, transfinite) & $0.005$ \\
Lemon slice & $[-2,2]\times[-2,2]$ & P1 (2D, unstructured triangular) & $0.1$ \\
Three-hole  & $[-2,2]\times[-1.5,2.5]$ & P1 (2D, unstructured triangular) & $0.1$ \\
\bottomrule
\end{tabular}
\caption{FEM reference solver: truncated domain, element type, and production mesh
size $h$ for each benchmark system, fixed after the mesh-refinement study of Section~\ref{sec:fem}.}
\label{tab:fem-domains}
\end{table}

\subsection{Robust Estimation of the Optimal Cost}\label{sec:opt-cost}
We first investigate the robustness of estimating the optimal cost $\lambda$ by solving the matrix eigenvalue problem~\eqref{eq:gev_generator}. For each $(\beta, \sigma, m)$ configuration, we run $20$ independent realizations and report the mean and standard deviation of the data-driven estimate, compared to the FEM reference, as shown in Fig.~\ref{fig:ME_lambda_RFF}.

\begin{figure}[h!]
    \centering
\includegraphics[width=\linewidth]{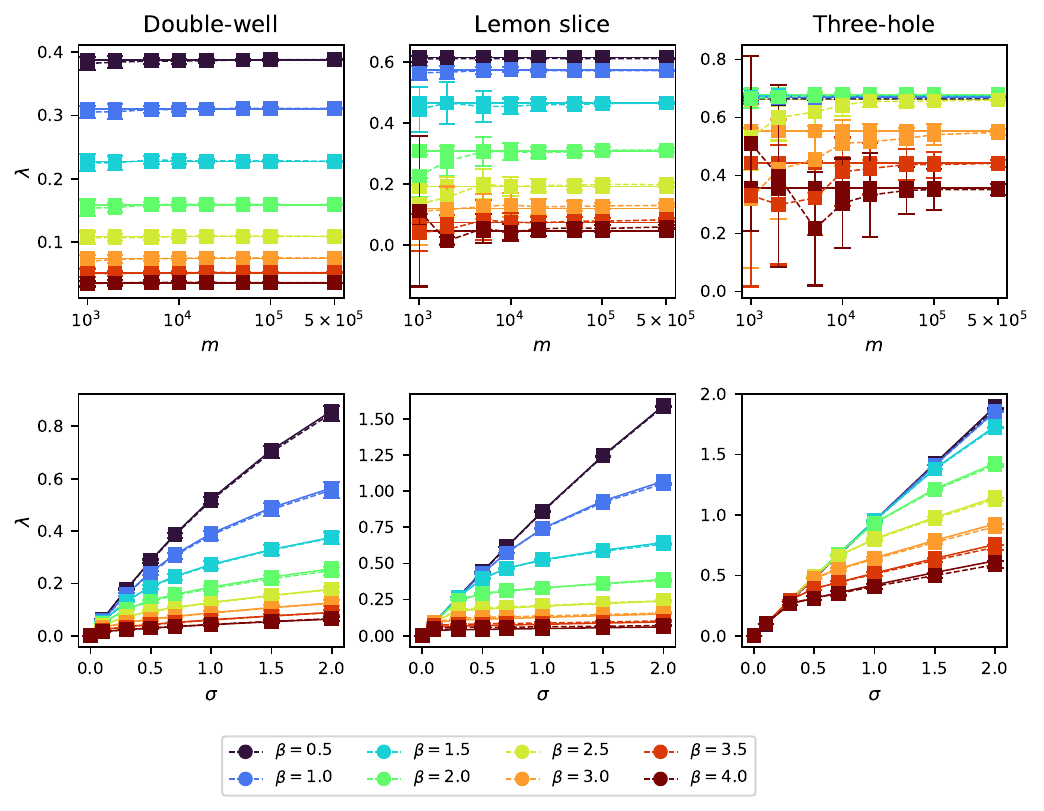}\caption{\textbf{Eigenvalue reconstruction} for the double-well (left), lemon slice (middle), and three-hole (right) systems. Top row: convergence with data size $m$ at fixed $\sigma=0.7$. Bottom row: dependence on the running cost prefactor $\sigma$ at fixed $m=2000$, $5\times10^{5}$, and $5\times10^{5}$ from left to right, respectively. Dashed lines with error bar indicate the RFF-gEDMD estimates (mean $\pm$ standard deviation over 20 independent realizations), and solid lines show the FEM reference. Colors indicate the inverse temperature $\beta$.}
    \label{fig:ME_lambda_RFF}
\end{figure}

Across all three systems and all $\beta$, the estimate stabilizes rapidly with $m$, with no systematic drift beyond $m\approx 2000$ for the double-well, $m\approx10^{4}$ for the lemon slice, and $m\approx10^{5}$ for the three-hole. For the latter system, only estimates at $\beta = 4.0$ converge slowly, which is the most extreme case. The variance across independent runs, quantified by the error bars, decreases monotonically with $m$ and becomes negligible for $m\ge10^{4}$. The gEDMD estimates thus closely track the FEM reference over the tested parameter range.

As a function of $\sigma$, at $\sigma\to0$, the operator $\cA$ approaches $-\cL$ and the leading eigenfunction approaches a constant, so the optimal cost approaches zero. Increasing $\sigma$ increases the energetic penalty associated with remaining outside the target and changes the principal eigen-pair accordingly. Results in Figure~\ref{fig:ME_lambda_RFF} are based on an RFF basis. We also repeated the
same calculation using the B-spline basis. The results are shown in Fig.~\ref{fig:ME_lambda_Bsp} in Appendix~\ref{sec:app_validation} and are nearly identical. Thus, the eigenvalue estimates are not strongly dependent on the particular basis used.

Beyond providing the optimal cost, the eigenvalue can also be used to recover equilibrium averages. As shown in Sec.~\ref{sec:equilibrium-average}, perturbation of the eigenvalue problem around $\sigma = 0$ yields
\begin{align*}
    \frac{\lambda(\sigma)}{\sigma} & =\mathbb{E}^{\mu}[\ell]+ \mathcal{O}(\sigma), \qquad \sigma \to 0^+.
\end{align*}
This provides a second use of the data-driven spectral approximation: once equilibrium training data are available, the eigenvalue problem can be evaluated repeatedly for small positive values of $\sigma$ without generating additional training data. Fig.~\ref{fig:Convergence} confirms that for each of the three model systems, $\lambda(\sigma)/\sigma$ reliably approaches the corresponding independently evaluated equilibrium average as $\sigma$ decreases. In these examples, the equilibrium average equals
\begin{equation*}
    \bE^\mu[\ell] =\bE^\mu[\one_{\bar{A}}] = 1 - \mu(A).
\end{equation*}

\begin{figure}[h!]
    \centering
\includegraphics[width=\linewidth]{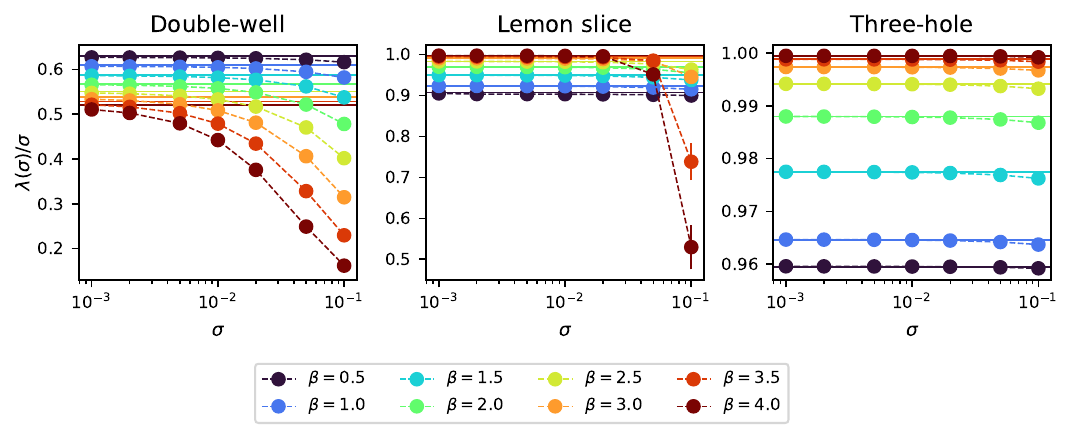}\caption{\textbf{Small-$\sigma$ convergence of the eigenvalue.} The ratio $\lambda(\sigma)/\sigma$ is shown as a function of $\sigma$ for the double-well (left), lemon slice (middle), and three-hole (right) systems. Dashed curves with markers show the RFF-gEDMD estimates (mean $\pm$ standard deviation over 20 independent realizations) at fixed $m=10^{4}$, $5\times10^{5}$, and $5\times10^{5}$ from left to right, respectively. The solid horizontal lines denote the corresponding equilibrium  averages $\mathbb{E}^{\mu}[\ell]$, evaluated independently from the equilibrium measure. Colors indicate the inverse temperature $\beta$. For the three-hole system, we extended the sampling domain to $\Omega=[-3,3]^2$ to account for the spread of the invariant measure at high temperature $\beta = 0.5$, using $n=200$ RFFs with $\gamma_{\mathrm{rff}}=0.5$.}
    \label{fig:Convergence}
\end{figure}

\subsection{Recovery of the Optimal Bias Potential}\label{sec:potentials}
Next, we examine the recovery and spatial structure of the optimal bias potential $V_{\mathrm{bias}}$ and the resulting optimal potential $V_{\mathrm{opt}} = V + V_{\mathrm{bias}}$. This comparison provides a more stringent test, since $V_{\mathrm{bias}} = -(2/\beta)\log \varphi$ depends directly on its spatial structure and is sensitive to small values of the eigenfunction $\varphi$.

For the double-well system (Fig.~\ref{fig:ME_dw_bias_Bsp}), the bias preferentially raises the competing right well relative to the target region in the left well. Consequently, $V_{\mathrm{opt}}$ develops a single dominant minimum at the target for all $\beta$. For sufficiently large $\sigma$, the competing minimum is effectively eliminated. The gEDMD estimates with a B-spline basis (Fig.~\ref{fig:ME_dw_bias_Bsp}) agree closely with the FEM reference over the full range of parameters shown. The corresponding RFF approximations are shown in Appendix~\ref{sec:app_validation} and are nearly identical.

\begin{figure}[h!]
    \centering
\includegraphics[width=0.8\linewidth]{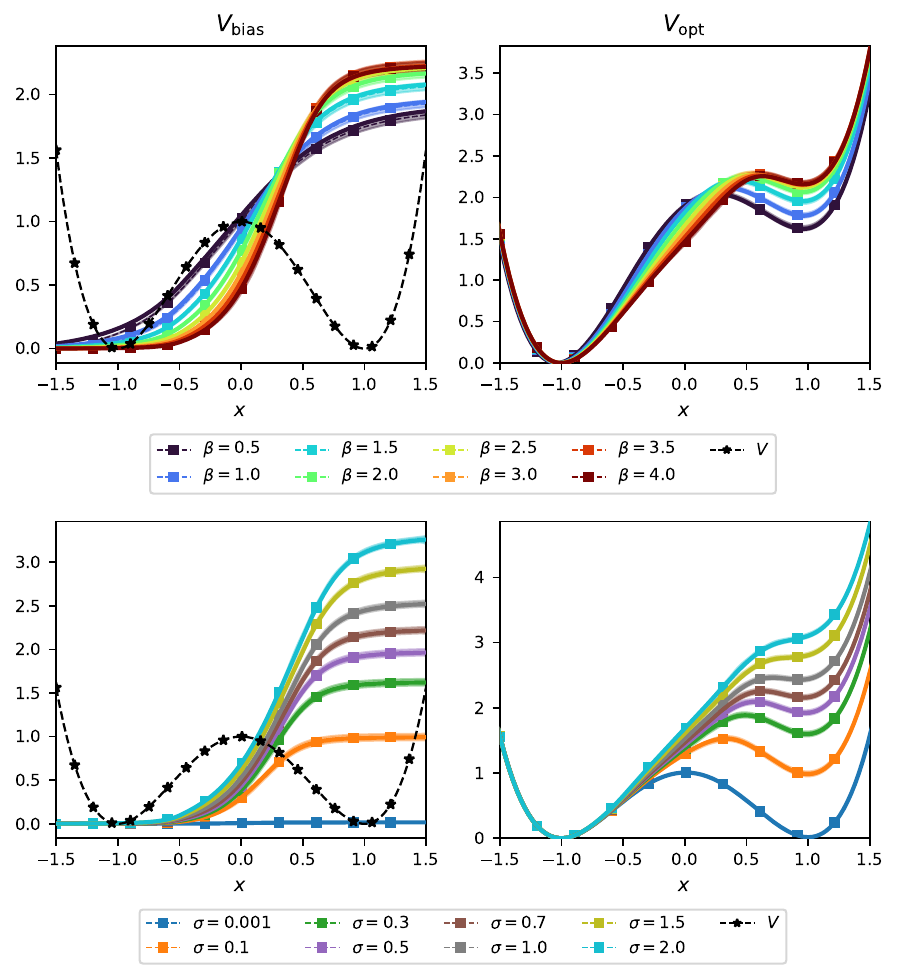}
    \caption{\textbf{Bias potential $V_{\mathrm{bias}}$ (left column) and optimal potential $V_{\mathrm{opt}}$} (right column) for the double-well system, obtained using B-spline gEDMD with $m=10^{4}$ samples. Top row: varying inverse temperature $\beta$ at fixed $\sigma=0.7$. Bottom row: varying running-cost prefactor $\sigma$ at fixed $\beta=4.0$. Dashed lines with markers show the mean over 20 independent realizations, with shaded bands indicating one standard deviation; solid lines of matching color show the FEM reference. The dashed black curve ($*$) shows the uncontrolled potential~$V$.}
    \label{fig:ME_dw_bias_Bsp}
\end{figure}

Similarly, the lemon slice system provides a two-dimensional test of the
reconstructed optimal landscape. Fig.~\ref{fig:ME_ls_Vopt} compares the FEM, RFF, and B-spline estimates of $V_{\mathrm{opt}}$ at fixed $\sigma=0.7$ for $\beta=1$ and $\beta=3$. In both cases, the optimal control reshapes the
landscape in favor of the target region in the upper-left part of the domain
relative to the competing wells. Both data-driven approximations reproduce
this global reorganization and the principal basin structure of the FEM
reference.

At $\beta=1$, the RFF and B-spline reconstructions are in good qualitative
agreement with the FEM solution, recovering the locations and relative
structure of the principal wells. At $\beta=3$, the target-directed reorganization remains clearly captured by both approximations, but differences in the reconstructed landscape become more visible. In particular, the RFF estimate exhibits stronger small-scale spatial oscillations, whereas the B-spline approximation gives a smoother
reconstruction while retaining the dominant basin structure. Thus, although
both representations recover the physically relevant reshaping of the
landscape, the B-spline basis provides a more regular reconstruction of $V_{\mathrm{opt}}$ in the low-temperature regime.

\begin{figure}[h!]
    \centering
\includegraphics[width=\linewidth]{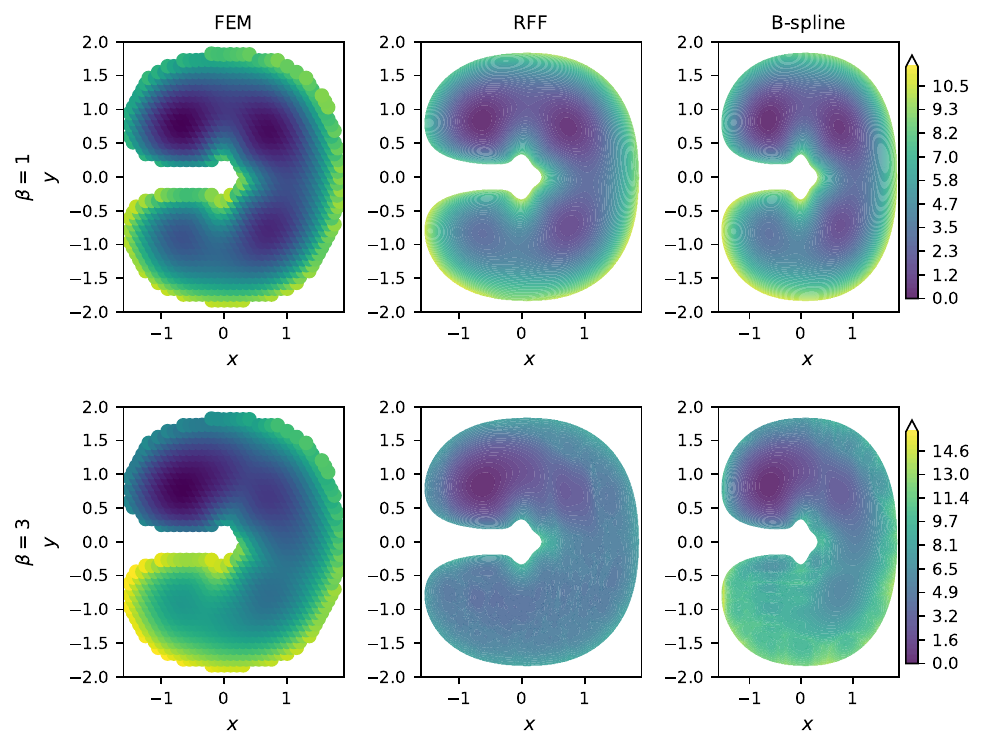}
    \caption{\textbf{Optimal potential $V_{\mathrm{opt}}$ for the lemon slice system} at fixed $\sigma=0.7$. Rows correspond to $\beta=1$ (top) and $\beta=3$ (bottom), and columns show the FEM reference, RFF-gEDMD reconstruction, and B-spline-gEDMD reconstruction, respectively. Within each row, all three methods are shown on the same color scale. Both data-driven reconstructions capture the target-favoring reorganization of the optimal landscape observed in the FEM reference.}
    \label{fig:ME_ls_Vopt}
\end{figure}

The three-hole system provides a more demanding two-dimensional example
because the prescribed target is the shallow upper well, while the uncontrolled
landscape contains two competing low-energy wells. Fig.~\ref{fig:ME_th_Vopt}
compares the FEM, RFF, and B-spline estimates of $V_{\mathrm{opt}}$ at fixed
$\sigma=2$ for $\beta=1$ and $\beta=2.5$. At $\beta=1$, the optimal control reorganizes the landscape so that the upper target basin is energetically
favored relative to the competing lower wells. Both RFF and B-spline reconstructions reproduce this overall structure and the locations of the principal basins observed in the FEM reference.

At $\beta=2.5$, the target-directed modification of the landscape remains
visible in both data-driven reconstructions. The principal upper basin and the structure of the competing lower region are retained, although local
deviations from the FEM reference become more apparent. The RFF estimate again displays somewhat stronger small-scale variations, while the B-spline representation produces a smoother landscape and more clearly preserves the
large-scale structure of the FEM solution. Despite these local differences,
both basis representations recover the essential effect of the optimal control: a reorganization of the physical energy landscape that favors the prescribed target region.

\begin{figure}[h!]
    \centering
\includegraphics[width=\linewidth]{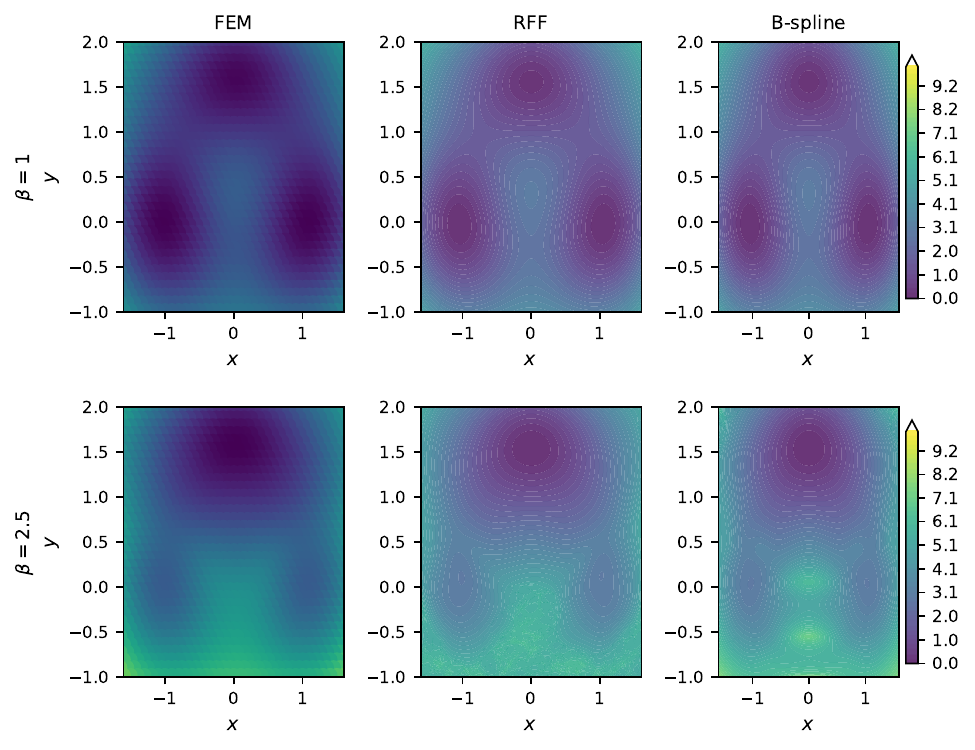}
    \caption{\textbf{Optimal potential $V_{\mathrm{opt}}$ for the three-hole system} at fixed $\sigma=2$. Rows correspond to $\beta=1$ (top) and $\beta=2.5$ (bottom), and columns show the FEM reference, RFF-gEDMD reconstruction, and B-spline-gEDMD reconstruction, respectively. Within each row, all three methods are shown on the same color scale. Both data-driven reconstructions capture the preferential stabilization of the upper target basin observed in the FEM reference.}
    \label{fig:ME_th_Vopt}
\end{figure}

We additionally solve the same variational problem using the neural Cole--Hopf representation introduced in Sec.~\ref{sec:nn}. In contrast to the linear basis expansions, the parametrization $\varphi_{\theta}=e^{-W_{\theta}}$ enforces positivity of the ground-state approximation by construction. Fig.~\ref{fig:nch} compares the NCH reconstruction with the FEM reference for the more challenging two-dimensional configurations. For the lemon slice system at $\beta=3$ and $\sigma=0.7$, the NCH reproduces the large-scale angular structure of optimal landscape and the preferential stabilization of the upper-left target basin. Similarly, for the three-hole system at $\beta=2.5$ and $\sigma=2$, the neural solution recovers the dominant target-directed reshaping of the landscape, with the upper target basin energetically favored relative to the lower region. In both cases, the NCH solutions extrapolate more smoothly to rarely sampled states than the linear basis expansions, but the relative weights of the competing states are not as accurate as for the gEDMD estimates, especially for the lemon slice case. These comparisons highlight an important distinction between estimating the principal eigenvalue and reconstructing the associated eigenfunction. While the eigenvalue is comparatively insensitive to the choice of approximation space, local errors in the eigenfunction can be amplified when constructing the bias potential through the logarithmic transformation. 

\begin{figure}[h!]
    \centering
    \includegraphics[width=0.8\linewidth]{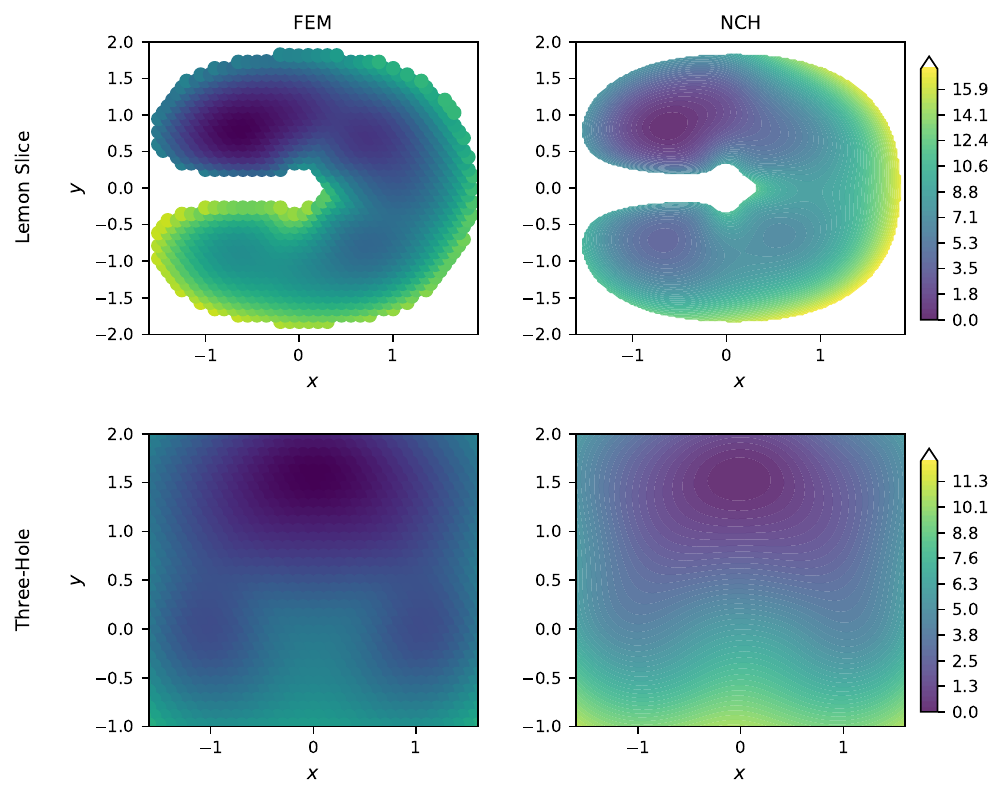}
    \caption{\textbf{Optimal potential $V_{\mathrm{opt}}$ from FEM reference (left) and neural Cole--Hopf (NCH) reconstruction (right)}. Top row: Lemon Slice system at $\beta=3$ and $\sigma=0.7$. Bottom row: Three-Hole system at $\beta=2.5$ and
    $\sigma=2$. The NCH results are averaged over five independent network initializations. The positivity-preserving neural parametrization
    reproduces the large-scale target-favoring reorganization of the optimal landscape while yielding a smooth reconstruction across the domain.}
    \label{fig:nch}
\end{figure}

\subsection{First-Passage-Time Validation}\label{sec:mfpt}
Upon visual inspection, the learned optimal potentials discussed in the previous section are expected to accelerate transitions into the target set under biased simulations. To quantify this effect, we compare the expected first-passage time (FPT) to the target set $A$ under the uncontrolled and controlled dynamics.  Starting from a fixed initial condition $x_0$ chosen within one of the competing states, we simulate $N_{\mathrm{traj}}=100$ independent trajectories of the biased dynamics (based on the B-spline approximation) using the Euler--Maruyama scheme with time step $\Delta t=10^{-3}$. For a single trajectory $\BX_t^k$, the first-passage time is defined as
$$
\tau(\BX_t^k) := \inf\{t\ge0 \mid \BX_{t}^k\in A\},
$$
and the mean first-passage time is estimated as
\begin{equation}
    \widehat{\mathrm{MFPT}}=\frac{1}{N_{\mathrm{traj}}}
    \sum_{k=1}^{N_{\mathrm{traj}}}\tau(\BX_t^k).
\end{equation}
For each controlled configuration $(\beta,\sigma)$, the calculation is repeated for 20 independently reconstructed controllers. We report the mean
of the resulting 20 MFPT estimates, with error bars indicating one standard deviation across the reconstructions. For the uncontrolled dynamics, a single ensemble of $N_{\mathrm{traj}}=100$ independent trajectories is used for each $\beta$. For the biased simulations, we limited the maximal runtime of each trajectory to a limit $t_{\max}$, which was chosen depending on $\beta$ based on the MFPT estimate for the unbiased dynamics. We have verified that all biased simulations reached the target set before time $t_{\max}$.

Figure~\ref{fig:mfpt} compares the mean first-passage time to the target set $A$ under the controlled and uncontrolled dynamics. The uncontrolled reference depends on $\beta$ but not on $\sigma$, whereas the controlled MFPT decreases as the state-cost prefactor is increased. For all three benchmark systems, the controlled MFPT quickly decreases compared to the uncontrolled value with increasing $\sigma$. The reduction is most pronounced at large $\beta$, where transitions under the uncontrolled dynamics become increasingly slow. In the high-$\beta$ regime, the uncontrolled MFPT reaches values on the order $O(10^2)$--$O(10^3)$, whereas the reconstructed controller reduces the transition time by up to three orders of magnitude. These observations confirm quantitatively that the learned optimal biasing potentials can significantly decrease the time it takes to reach the target set, as intended.

\begin{table}[t]
\centering
\begin{tabular}{lcc}
\toprule
System & $x_0$ & $t_{\max}(\beta)$ \\
\midrule
Double-well & $1.0$ & $[15, 30, 60, 90, 150, 220]$ \\
Lemon slice & $(0.7, 0.7)$ & $[45, 135, 435, 1425, 4950, 11400]$ \\
Three-hole & $(-1.0, 0.0)$ & $[40, 90, 240, 750, 2200, 7450]$ \\
\bottomrule
\end{tabular}
\caption{Settings for mean first-passage-time calculations. For each system, the maximum simulation horizon is chosen depending on $\beta$ based on the MFPT of the physical dynamics. The entries
in the third column correspond to $\beta=\{0.5,1.0,1.5,2.0,2.5,3.0\}$. For all reported simulations, we verified that all trajectories reached the target within the prescribed simulation horizon.}
\label{tab:mfpt}
\end{table}

\begin{figure}[h!]
    \centering
\includegraphics[width=\linewidth]{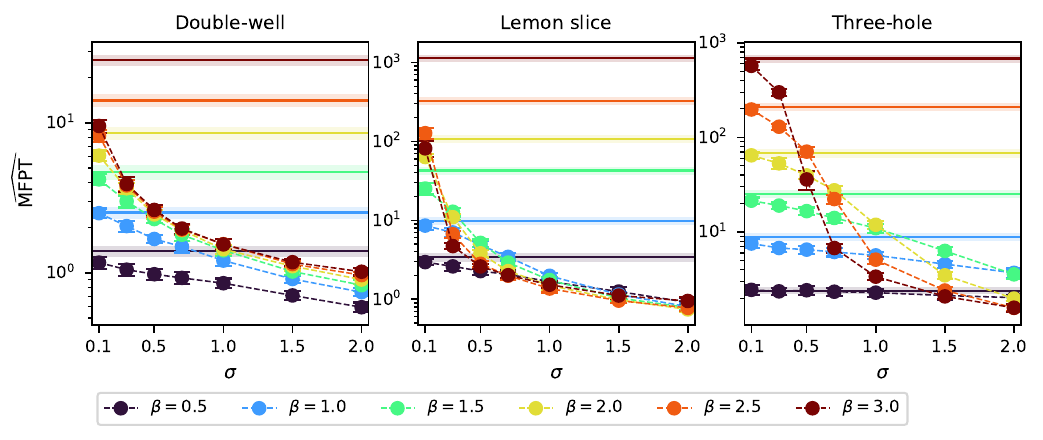}
    \caption{\textbf{Mean first-passage time to the target set $A$} for the double-well (left), Lemon-slice (middle), and three-hole (right) systems. Markers show the mean controlled MFPT over 20 independently reconstructed controllers, with error bars denoting one standard deviation across the reconstructions. Each controller-specific MFPT is estimated from $N_{\mathrm{traj}}=100$ trajectories. Solid horizontal lines show the
    uncontrolled MFPT obtained from $100$ independent trajectories, and the corresponding shaded bands indicate the standard error of the mean. Colors denote the inverse temperature $\beta$.}
    \label{fig:mfpt}
\end{figure}

\subsection{Temperature Re-weighting}\label{sec:reweighting-results}
We finally assess whether a single inexpensive simulation at a higher reference temperature can provide accurate eigen-pair estimates at lower target
temperatures. To this end, we choose $\beta_0 = 0.1$ and generate samples from the physical dynamics at this reference temperature. The samples are then re-weighted according to~\eqref{eq:temp_reweight} before assembling the gEDMD matrices. We set $(m, \gamma_{\mathrm{rff}})$ to $(5\times10^{3},0.1)$ for the double-well system, $(10^{4},0.25)$ for the lemon slice system and $(5\times10^{4},0.3)$ for the three-hole system. 
Across all three benchmark systems, the reweighted eigenvalue estimates (dashed lines) closely track the FEM reference (solid lines) over the full range of target inverse temperatures $\beta\in[0.5, 4.0]$, as shown in Fig.~\ref{fig:reweight}. To relate this accuracy to the actual sampling effort, we compare the transition time $\tau_{\mathrm{K}}(\beta)$ of the uncontrolled dynamics at the target inverse temperature with the physical duration $T_{\mathrm{data}}(\beta_0):=m\times\Delta t$ of the reference-temperature trajectory used to generate the data. The transition times $\tau_{\mathrm{K}}(\beta)$ are estimated from the Eyring-Kramers law, see~\cite{lelievre_partial_2016}. We define the dimensionless ratio $R_{\mathrm{data}}(\beta):=\frac{\tau_{\mathrm{K}}(\beta)}{T_{\mathrm{data}}(\beta_0)}$. Thus, $R_{\mathrm{data}}>1$ indicates that the trajectory used for the reweighted estimate is shorter than a single transition time of the uncontrolled dynamics at the target temperature.

As shown in the right panel of Fig.~\ref{fig:reweight}, $R_{\mathrm{data}}$ increases rapidly with $\beta$. At the largest inverse temperature it exceeds the duration of the reference-temperature trajectory by approximately one to three orders of magnitude, depending on the system. Nevertheless, the corresponding reweighted eigenvalue estimates remain close to the FEM reference. This demonstrates that, for these low-dimensional benchmark systems, temperature reweighting can recover the spectral quantity of interest from reference-temperature data generated over a time interval substantially shorter than the transition timescale of the target dynamics.

We emphasize that $R_{\mathrm{data}}$ is a timescale-based measure of sampling advantage rather than an exact computational speedup. The Eyring--Kramers time characterizes a typical barrier-crossing timescale, whereas statistically converged direct sampling at the target temperature may require multiple transitions.

\begin{figure}[h!]
    \centering
\includegraphics[width=\linewidth]{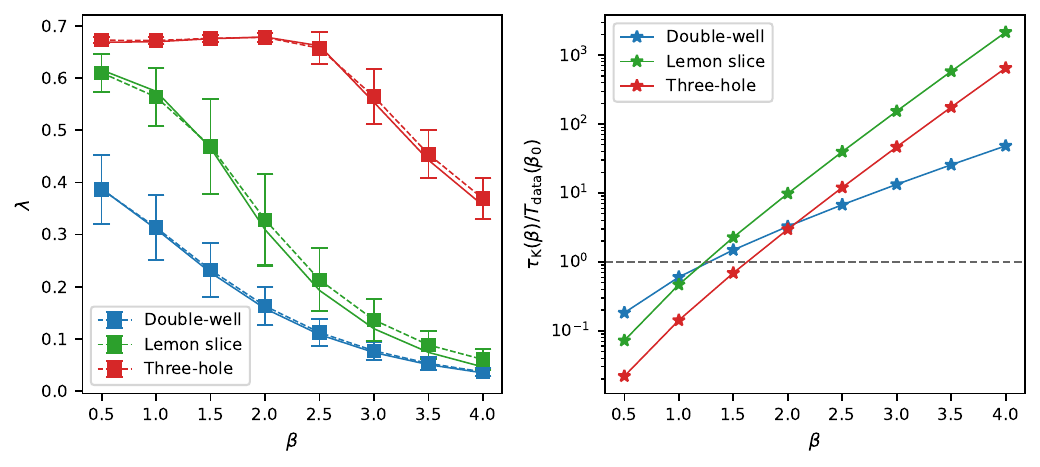}
    \caption{\textbf{Temperature re-weighting} from a single reference simulation at $\beta_{0}=0.1$ across the three benchmark systems. Left: eigenvalue $\lambda$ obtained from reweighted samples using RFF-gEDMD (dashed lines with error bars) compared with the FEM reference (solid lines). Right: ratio $\frac{\tau_{\mathrm{K}}(\beta)}{T_{\mathrm{data}}(\beta_0)}$ between the Eyring--Kramers transition time at the target inverse temperature and the physical duration of the reference-temperature
    trajectory. The horizontal line at unity indicates $\tau_{\mathrm K}(\beta)=T_{\mathrm{data}}(\beta_0)$.}
    \label{fig:reweight}
\end{figure}

\section{Application to Alanine Dipeptide}
\label{sec:ad}
In this section, we apply the proposed framework to a coarse-grained representation of the classical alanine dipeptide benchmark. We consider MD simulations of alanine dipeptide in explicit water, with the force field and simulation parameters described in Ref.~\cite{nateghi_kinetically_2025}. A $500~\mathrm{ns}$ unbiased simulation serves as reference data set. In addition, we also produced a $20\,\mathrm{ns}$ long well-tempered metadynamics simulation using the backbone dihedral angles $(\phi, \psi)$ as
collective variables. Metadynamics simulations were carried out in \textsc{Gromacs}~\cite{abraham_gromacs_2015} with the \textsc{Plumed} plug-in~\cite{the_plumed_consortium_promoting_2019}. Figure~\ref{fig:ad_energy} shows the free energy landscape in the dihedral angle space as estimated from both datasets, indicating good agreement.

The target set for the bias potential is chosen as the shallow energetic minimum in the right half of the dihedral plane, defined as
\begin{equation}
    A:=\{\,(\phi,\psi)\mid |\phi-0.9|<0.3\,\}.
\end{equation}


\begin{figure}[h!]
    \centering
    \includegraphics[width=\linewidth]{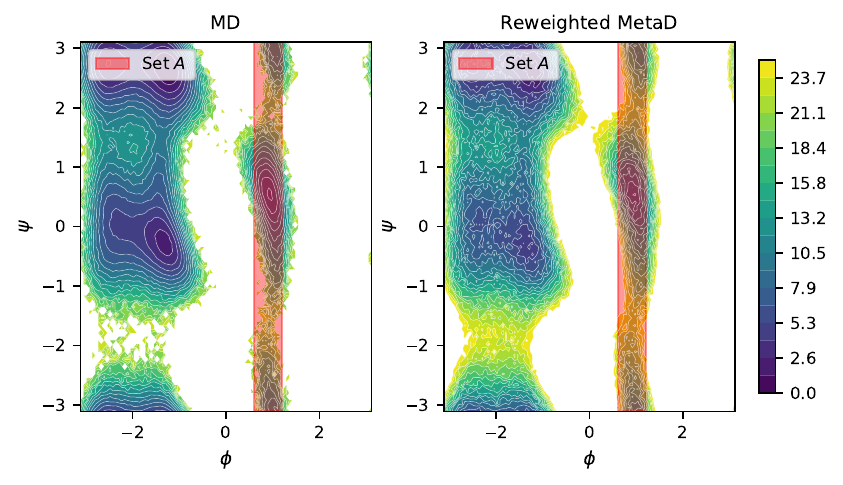}
    \caption{\textbf{Free energy of alanine dipeptide} estimated from unbiased MD trajectory (left panel) and MetaD re-weighted trajectory (right panel). The target set for the optimal control problem is indicated by the red strip.}
    \label{fig:ad_energy}
\end{figure}

 \subsection{Bias Potential and Coarse-Grained Generator}
 Before presenting numerical results, we explore an interesting connection between the variational approximation~\eqref{eq:ev_problem} and coarse-graining of the original dynamics.

 Consider a mapping from the Euclidean state space $\R^d$ into a lower-dimensional space $\R^q$, that is $\xi: \R^d \to  \R^q$ and write $\bz = \xi(\bx)$. Given the full space dynamics~\eqref{eq:sde_langevin} with generator~ $\cL$, one can define a closed Markovian dynamics in $\R^q$ by projecting the generator:
 \begin{equation}
     \cL^\xi = \cP \cL \cP,
 \end{equation}
 where $\cP$ is the Mori-Zwanzig projector. The projected generator $\cL^\xi$ gives rise to stochastic dynamics on $\R^q$:
 \begin{equation}
     \label{eq:cg_dynamics}
     \diff \BZ_t = \left[-\ba(\BZ_t)\nabla F(\BZ_t) + \frac{1}{\beta}\nabla \cdot \ba(\BZ_t) \right] \,\diff t + \sqrt{2\beta^{-1}}\ba(\BZ_t)^{1/2} \,\diff \BB_t.
 \end{equation}
 Here, $F$ is the free energy in $\R^q$ and $\ba$ is the effective diffusion field. For details on this construction, see Refs.~\cite{legollEffectiveDynamicsUsing2010a,zhang_effective_2016,nateghi_kinetically_2025}. Also note that the effective dynamics~\eqref{eq:cg_dynamics} ignores any memory effects due to the coarse-graining.

 It was shown in~\cite{zhang_effective_2016} that any Galerkin-type approximation to $\cL$ using a basis set \emph{composed of functions on $\bz$} is also a Galerkin approximation to the projected generator $\cL^\xi$. The significance of this result is the following: if
 \begin{enumerate}
     \item we choose to represent the eigenfunction $\varphi$ in~\eqref{eq:ev_problem} and hence the optimal bias $\vbias$ by functions on $\bz$:
 \begin{equation*}
     \varphi(\bx) \approx \varphi(\xi(\bx)) \approx \sum_{j=1}^n \bv^m_j \psi_j(\xi(\bx)) = \sum_{j=1}^n \bv^m_j \psi_j(\bz);
 \end{equation*}
    \item the running cost $\ell$ is also a function of $\bz$;
 \end{enumerate}
 then the generalized matrix eigenvalue problem~\eqref{eq:ev_problem} approximates the optimal control problem for the coarse-grained generator:
 \begin{equation}
 \label{eq:ev_problem_cg}
     \cA^\xi \varphi = (\sigma \ell - \cL^\xi)\varphi = \lambda \varphi.
 \end{equation}
 The resulting bias potential $\vbias$ can be used to guide the projected dynamics~\eqref{eq:cg_dynamics}:
 \begin{equation}
     \label{eq:cg_biased_dynamics}
     \diff \BZ_t = \left[-\ba(\BZ_t)(\nabla F(\BZ_t) + \nabla \vbias(\BZ_t)) + \frac{1}{\beta}\nabla \cdot \ba(\BZ_t) \right] \,\diff t + \sqrt{2\beta^{-1}}\ba(\BZ_t)^{1/2} \,\diff \BB_t.
 \end{equation}

\subsection{Spectral and Potential Comparison}
With these considerations in mind, we apply the gEDMD approximation to approximate the eigenvalue $\lambda$ in~\eqref{eq:ev_problem_cg} using a basis set of random features on the $(\phi, \psi)$-dihedral space. We compare the spectral quantities obtained from the equilibrium MD and re-weighted MetaD data. Fig.~\ref{fig:ad_spectral} summarizes this comparison together with the resulting bias and optimal potentials. 

Panel~(a) shows the principal eigenvalue $\lambda$ as a function of the running cost prefactor $\sigma$. The two estimates exhibit consistent dependence on $\sigma$ over the range considered. Note that, as shown in the inset, the limiting slope of the principal eigenvalue $\lambda(\sigma)$ (dashed lines) as $\sigma\to 0$ approximates the equilibrium probability $\mathbb{E}^{\mu}[\ell(\bx)] = 1-\mu(A)$. Table~\ref{tab:ad_equilibrium} reports a quantitative comparison of $\lambda^{\prime}(0)$ estimated from the gEDMD eigenvalues using the second-order forward finite-difference 
$$
\lambda^{\prime}(0) \approx \frac{4\lambda(h)-\lambda(2h)}{2h}, \qquad h=10^{-3},
$$
compared to a direct estimate obtained either from the unbiased MD data or metadynamics simulations.

For $\sigma=1.0$, panels~(b) and~(c) show the bias potentials $V_{\mathrm{bias}}$ reconstructed from the MD and MetaD data, respectively. Both estimates are found to be in agreement. The corresponding effective controlled potentials $V_{\mathrm{opt}}$ are shown in panels~(d) and~(e). We can see that the bias potential significantly lowers the free energy in the right-hand minimum corresponding to the target state relative to the competing minima, to facilitate transitions. We next examine whether this agreement also leads to comparable acceleration of the closed-loop dynamics.

\begin{figure}
    \centering
    \includegraphics[width=\linewidth]{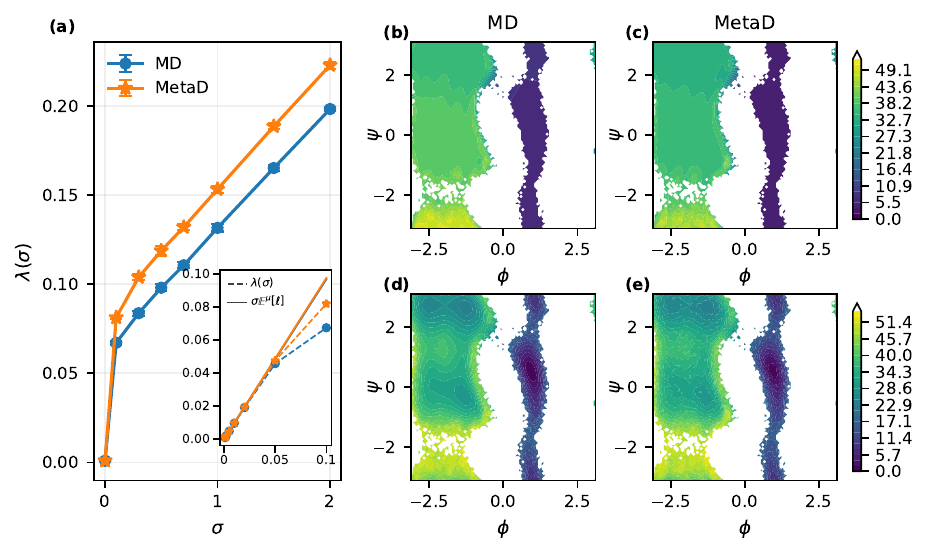}
    \caption{\textbf{Comparison of the spectral quantities for alanine dipeptide}: (a) Eigenvalue $\lambda(\sigma)$ from MD and reweighted MetaD data. The inset magnifies the small-$\sigma$ regime; dashed curves with markers denote the spectral estimates, while solid lines show the linear extrapolation $\sigma\mathbb{E}^{\mu}[\ell]$. (b,c) Bias potentials $V_{\mathrm{bias}}$ for $\sigma=1.0$ obtained from MD and MetaD data, respectively. (d,e) Corresponding optimal potentials $V_{\mathrm{opt}}$. Within each row, MD and MetaD use the same color scale.}
    \label{fig:ad_spectral}
\end{figure}

\begin{table}[t]
    \centering
    \begin{tabular}{lccc}
        \toprule
        Dataset & Full-data $\mu(A)$ & Full-data $\mathbb{E}^{\mu}[\ell]$ &$\lambda^{\prime}(0)$\\
        \midrule
        MD & $0.03145$ & $0.96855$ & $0.96856 \pm 8.3\times10^{-7}$\\
        MetaD & $0.02393$ & $0.97607$ & $0.97608 \pm 4.9\times10^{-7}$ \\
        \bottomrule
    \end{tabular}
\caption{Equilibrium target probabilities and small-$\sigma$ spectral
    estimates for alanine dipeptide. The reference target probability $\mu(A)$ is computed from the full dataset, with $\mathbb{E}^{\mu}[\ell]=1-\mu(A)$. The derivative $\lambda^{\prime}(0)$ is estimated using the second-order forward finite difference with $h=10^{-3}$ and is reported as the mean $\pm$ standard deviation
    over 20 RFF realizations.}
    \label{tab:ad_equilibrium}
\end{table}

\subsection{Closed-Loop Validation}
Closed-loop trajectories are generated using the effective dynamics~\eqref{eq:cg_biased_dynamics} in the space of collective variables $\bz = (\phi,\psi)$. As before, we quantify the acceleration in terms of the mean first-passage time (MFPT) to the set $A$. Starting from the same initial condition, we simulate $N_{\mathrm{traj}} = 100$ independent trajectories for each reconstructed controller using a time step $\Delta t=10^{-3}\,\mathrm{ps}$. The MFPT is computed using the same first-passage criterion as in Sec.~\ref{sec:mfpt} and shown in Fig.~\ref{fig:mfpt_ad}. For each value of $\sigma$, the reported controlled MFPT is the mean over the independently reconstructed eigen-pairs, with error bars indicating the corresponding standard deviation. The uncontrolled reference is estimated directly from an ensemble of $N_{\mathrm{traj}}=100$ trajectories; the shaded region indicates the standard error of the estimated uncontrolled MFPT. Note that time is re-scaled and accelerated for the coarse-grained dynamics~\eqref{eq:cg_dynamics}, the corresponding MFPT estimate based on the all-atom simulation data is $\approx 35\,\mathrm{ns}$.

\begin{figure}
    \centering
    \includegraphics[width=\linewidth]{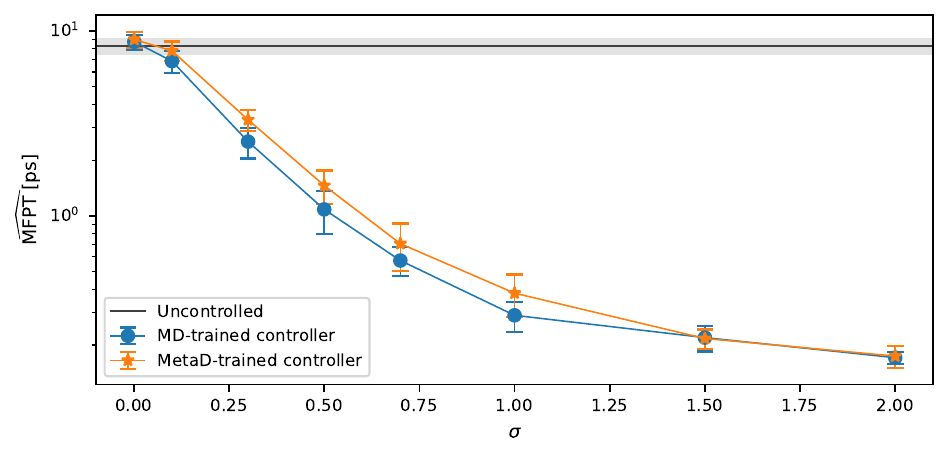}
    \caption{\textbf{Mean first-passage time (MFPT) for alanine dipeptide} under controllers reconstructed from equilibrium MD data (circles) and reweighted MetaD data (stars), as a function of the running-cost prefactor $\sigma$. Controlled MFPTs are computed from $N_{\mathrm{traj}}=100$ trajectories for each reconstructed eigen-pair; markers show the mean over 20 independently reconstructed eigen-pairs and error bars denote one standard deviation across reconstructions. The horizontal line shows the uncontrolled MFPT estimated from $N_{\mathrm{traj}}=100$ trajectories, and the shaded band denotes the standard error of this mean. The logarithmic vertical scale highlights the systematic reduction in transition time with increasing $\sigma$.}
    \label{fig:mfpt_ad}
\end{figure}

We observe that both the MD- and MetaD-trained controllers substantially reduce the MFPT relative to the uncontrolled dynamics. As $\sigma$ increases, the MFPT decreases systematically, consistent with the increasing strength of the state penalty and the resulting control bias. Importantly, the controllers reconstructed from equilibrium MD and reweighted MetaD data exhibit closely consistent first-passage behavior over the full range of $\sigma$ considered. 


\section{Conclusions}
We have studied the calculation of optimal biasing potentials to accelerate metastable transitions and compute equilibrium expectations in molecular dynamics simulations. Starting from an ergodic optimal control problem, we have exploited that the OCP can be reformulated as a linear eigenvalue problem for the generator of the physical dynamics. We have formulated data-driven learning methods to solve this eigenvalue problem, and to determine equilibrium expectations by extrapolating the slope of these eigenvalues at $\sigma = 0$. We have studied in detail the robustness and accuracy of these learning methods using both model problems and MD simulations of the alanine dipeptide. Finally, we have also analyzed the relation of the ergodic OCP to coarse-grained dynamics and generator approximation on reaction coordinates.

As already noted in the introduction, future work will focus on lifting the need for equilibrium samples, on extending the scope of the method to larger systems, and on applying the method in an adaptive sampling context.

\section*{Declarations}
The authors have no conflicts of interest to disclose.

\paragraph{Data Availability}
Data and codes to reproduce the results presented in this study are available from the public repository hosted on \textsc{Zenodo} at \url{10.5281/zenodo.22070320}.

\paragraph{Use of Artificial Intelligence} During the preparation of this work, \textsc{Claude Sonnet 4.6} and \textsc{ChatGPT Sol 5.6} were used as writing aids, primarily for language refinement, as well as for code polishing and debugging. All mathematical derivations, numerical results, and scientific conclusions are those of the authors. The authors independently verified all AI-assisted code and reviewed all AI-assisted content, and remain fully responsible for the content of the manuscript. 

\paragraph{Funding and Acknowledgemts} L.G. is funded by Deutsche Forschungsgemeinschaft (DFG, German Research Foundation) - 314838170, GRK 2297 MathCoRe. T.B. acknowledges funding by the Deutsche Forschungsgemeinschaft (DFG, German Research Foundation) -- Project-ID 544702565. The work C.H. has been partially funded by the German Federal Government, the Federal Ministry of Research, Technology and Space, and the State of Brandenburg within the framework of the joint project EIZ: Energy Innovation Center (project numbers 85056897 and 03SF0693A) with funds from the Structural Development Act
(Strukturstärkungsgesetz) for coal-mining regions.

\bibliographystyle{unsrt}
\bibliography{expEMUS}

@book{arapostathis2012ergodic,
  title={Ergodic control of diffusion processes},
  author={Arapostathis, Ari and Borkar, Vivek S and Ghosh, Mrinal K},
  number={143},
  year={2012},
  publisher={Cambridge University Press}
}

@book{frenkelUnderstandingMolecularSimulation2023,
    address = {S.l.},
    edition = {Third edition},
    title = {Understanding molecular simulation: from algorithms to applications},
    isbn = {978-0-323-91318-8},
    shorttitle = {Understanding molecular simulation},
    language = {eng},
    publisher = {ELSEVIER ACADEMIC PRESS},
    author = {Frenkel, Daan and Smit, Berend},
    year = {2023},
}

@article{torrie_nonphysical_1977,
    title = {Nonphysical sampling distributions in {Monte} {Carlo} free-energy estimation: {Umbrella} sampling},
    volume = {23},
    number = {2},
    journal = {Journal of computational physics},
    author = {Torrie, Glenn M and Valleau, John P},
    year = {1977},
    pages = {187--199},
}

@article{laio_escaping_2002,
    title = {Escaping free-energy minima},
    volume = {99},
    issn = {0027-8424},
    doi = {10.1073/pnas.202427399},
    number = {20},
    journal = {Proceedings of the National Academy of Sciences},
    author = {Laio, Alessandro and Parrinello, Michele},
    month = oct,
    year = {2002},
    pages = {12562--12566},
}

@article{swendsen_replica_1986,
    title = {Replica {Monte} {Carlo} {Simulation} of {Spin}-{Glasses}},
    volume = {57},
    url = {https://link.aps.org/doi/10.1103/PhysRevLett.57.2607},
    doi = {10.1103/PhysRevLett.57.2607},
    number = {21},
    urldate = {2025-01-28},
    journal = {Physical Review Letters},
    author = {Swendsen, Robert H. and Wang, Jian-Sheng},
    month = nov,
    year = {1986},
    pages = {2607--2609},
}

@article{voter_parallel_1998,
    title = {Parallel replica method for dynamics of infrequent events},
    volume = {57},
    copyright = {http://link.aps.org/licenses/aps-default-license},
    issn = {0163-1829, 1095-3795},
    url = {https://link.aps.org/doi/10.1103/PhysRevB.57.R13985},
    doi = {10.1103/PhysRevB.57.R13985},
    language = {en},
    number = {22},
    urldate = {2026-08-18},
    journal = {Physical Review B},
    author = {Voter, Arthur F.},
    month = jun,
    year = {1998},
    pages = {R13985--R13988},
}

@article{darve_calculating_2001,
    title = {Calculating free energies using average force},
    volume = {115},
    issn = {0021-9606, 1089-7690},
    url = {https://pubs.aip.org/jcp/article/115/20/9169/442127/Calculating-free-energies-using-average-force},
    doi = {10.1063/1.1410978},
    language = {en},
    number = {20},
    urldate = {2026-08-18},
    journal = {The Journal of Chemical Physics},
    author = {Darve, Eric and Pohorille, Andrew},
    month = nov,
    year = {2001},
    pages = {9169--9183},
}

@article{schutte_optimal_2012,
    title = {Optimal control of molecular dynamics using {Markov} state models},
    volume = {134},
    copyright = {http://www.springer.com/tdm},
    issn = {0025-5610, 1436-4646},
    url = {http://link.springer.com/10.1007/s10107-012-0547-6},
    doi = {10.1007/s10107-012-0547-6},
    language = {en},
    number = {1},
    urldate = {2026-08-18},
    journal = {Mathematical Programming},
    author = {Schütte, Christof and Winkelmann, Stefanie and Hartmann, Carsten},
    month = aug,
    year = {2012},
    pages = {259--282},
}

@inproceedings{holdijk_stochastic_2023,
    address = {New Orleans, Louisiana, USA},
    title = {Stochastic {Optimal} {Control} for {Collective} {Variable} {Free} {Sampling} of {Molecular} {Transition} {Paths}},
    isbn = {9781713899112},
    url = {http://www.proceedings.com/075280-3481.html},
    doi = {10.52202/075280-3481},
    urldate = {2026-08-18},
    booktitle = {Advances in {Neural} {Information} {Processing} {Systems} 36},
    publisher = {Neural Information Processing Systems Foundation, Inc. (NeurIPS)},
    author = {Holdijk, Lars and Du, Yuanqi and Hooft, Ferry and Jaini, Priyank and Ensing, Berend and Welling, Max},
    year = {2023},
    pages = {79540--79556},
}

@article{hartmann_efficient_2012,
    title = {Efficient rare event simulation by optimal nonequilibrium forcing},
    volume = {2012},
    issn = {1742-5468},
    url = {https://iopscience.iop.org/article/10.1088/1742-5468/2012/11/P11004},
    doi = {10.1088/1742-5468/2012/11/P11004},
    number = {11},
    urldate = {2026-08-18},
    journal = {Journal of Statistical Mechanics: Theory and Experiment},
    author = {Hartmann, Carsten and Schütte, Christof},
    month = nov,
    year = {2012},
    pages = {P11004},
}

@article{vandeneijnden_rare_2012,
    title = {Rare {Event} {Simulation} of {Small} {Noise} {Diffusions}},
    volume = {65},
    issn = {0010-3640, 1097-0312},
    url = {https://onlinelibrary.wiley.com/doi/10.1002/cpa.21428},
    doi = {10.1002/cpa.21428},
    language = {en},
    number = {12},
    urldate = {2026-08-18},
    journal = {Communications on Pure and Applied Mathematics},
    author = {Vanden‐Eijnden, Eric and Weare, Jonathan},
    month = dec,
    year = {2012},
    pages = {1770--1803},
}

@article{klus_data-driven_2020,
    title = {Data-driven approximation of the {Koopman} generator: {Model} reduction, system identification, and control},
    volume = {406},
    doi = {10.1016/j.physd.2020.132416},
    journal = {Physica D: Nonlinear Phenomena},
    author = {Klus, Stefan and Nüske, Feliks and Peitz, Sebastian and Niemann, Jan Hendrik and Clementi, Cecilia and Schütte, Christof},
    year = {2020},
}

@article{nuske_efficient_2023,
    title = {Efficient approximation of molecular kinetics using random {Fourier} features},
    volume = {159},
    issn = {0021-9606, 1089-7690},
    url = {https://pubs.aip.org/jcp/article/159/7/074105/2907133/Efficient-approximation-of-molecular-kinetics},
    doi = {10.1063/5.0162619},
    language = {en},
    number = {7},
    urldate = {2025-01-09},
    journal = {The Journal of Chemical Physics},
    author = {Nüske, Feliks and Klus, Stefan},
    month = aug,
    year = {2023},
    pages = {074105},
}

@article{rahimiRandomFeaturesLargescale2007a,
    title = {Random features for large-scale kernel machines},
    volume = {20},
    journal = {Advances in Neural Information Processing Systems},
    author = {Rahimi, A. and Recht, B.},
    year = {2007},
}

@article{geuzaine_gmsh_2009,
    title = {Gmsh: {A} 3‐{D} finite element mesh generator with built‐in pre‐ and post‐processing facilities},
    volume = {79},
    copyright = {http://onlinelibrary.wiley.com/termsAndConditions\#vor},
    issn = {0029-5981, 1097-0207},
    shorttitle = {Gmsh},
    url = {https://onlinelibrary.wiley.com/doi/10.1002/nme.2579},
    doi = {10.1002/nme.2579},
    language = {en},
    number = {11},
    urldate = {2026-08-19},
    journal = {International Journal for Numerical Methods in Engineering},
    author = {Geuzaine, Christophe and Remacle, Jean‐François},
    month = sep,
    year = {2009},
    pages = {1309--1331},
}

@misc{baratta_dolfinx_2023,
    title = {{DOLFINx}: {The} next generation {FEniCS} problem solving environment},
    copyright = {Creative Commons Attribution 4.0 International},
    shorttitle = {{DOLFINx}},
    url = {https://zenodo.org/doi/10.5281/zenodo.10447666},
    doi = {10.5281/ZENODO.10447666},
    language = {en},
    urldate = {2026-08-19},
    publisher = {Zenodo},
    author = {Baratta, Igor A. and Dean, Joseph P. and Dokken, Jørgen S. and Habera, Michal and Hale, Jack S. and Richardson, Chris N. and Rognes, Marie E. and Scroggs, Matthew W. and Sime, Nathan and Wells, Garth N.},
    month = dec,
    year = {2023},
}

@article{lelievre_partial_2016,
    title = {Partial differential equations and stochastic methods in moleculardynamics},
    volume = {25},
    issn = {0962-4929, 1474-0508},
    url = {https://www.cambridge.org/core/journals/acta-numerica/article/abs/partial-differential-equations-and-stochastic-methods-in-moleculardynamics/60F8398275D5150AA54DD98F745A9285},
    doi = {10.1017/S0962492916000039},
    language = {en},
    urldate = {2026-01-07},
    journal = {Acta Numerica},
    author = {Lelièvre, Tony and Stoltz, Gabriel},
    month = may,
    year = {2016},
    pages = {681--880},
}

@article{nateghi_kinetically_2025,
    title = {Kinetically {Consistent} {Coarse} {Graining} {Using} {Kernel}-{Based} {Extended} {Dynamic} {Mode} {Decomposition}},
    volume = {21},
    copyright = {https://creativecommons.org/licenses/by/4.0/},
    issn = {1549-9618, 1549-9626},
    url = {https://pubs.acs.org/doi/10.1021/acs.jctc.5c00479},
    doi = {10.1021/acs.jctc.5c00479},
    language = {en},
    number = {15},
    urldate = {2025-10-08},
    journal = {Journal of Chemical Theory and Computation},
    author = {Nateghi, Vahid and Nüske, Feliks},
    month = aug,
    year = {2025},
    pages = {7236--7248},
}

@article{the_plumed_consortium_promoting_2019,
    title = {Promoting transparency and reproducibility in enhanced molecular simulations},
    volume = {16},
    issn = {1548-7091, 1548-7105},
    url = {https://www.nature.com/articles/s41592-019-0506-8},
    doi = {10.1038/s41592-019-0506-8},
    language = {en},
    number = {8},
    urldate = {2026-08-20},
    journal = {Nature Methods},
    author = {{The PLUMED consortium}},
    month = aug,
    year = {2019},
    pages = {670--673},
}

@article{abraham_gromacs_2015,
    title = {{GROMACS}: {High} performance molecular simulations through multi-level parallelism from laptops to supercomputers},
    volume = {1-2},
    issn = {23527110},
    shorttitle = {{GROMACS}},
    url = {https://linkinghub.elsevier.com/retrieve/pii/S2352711015000059},
    doi = {10.1016/j.softx.2015.06.001},
    language = {en},
    urldate = {2026-08-20},
    journal = {SoftwareX},
    author = {Abraham, Mark James and Murtola, Teemu and Schulz, Roland and Páll, Szilárd and Smith, Jeremy C. and Hess, Berk and Lindahl, Erik},
    month = sep,
    year = {2015},
    pages = {19--25},
}

@article{legollEffectiveDynamicsUsing2010a,
    title = {Effective dynamics using conditional expectations},
    volume = {23},
    issn = {0951-7715},
    doi = {10.1088/0951-7715/23/9/006},
    number = {9},
    journal = {Nonlinearity},
    author = {Legoll, Frédéric and Lelièvre, Tony},
    month = sep,
    year = {2010},
    pages = {2131--2163},
}

@article{zhang_effective_2016,
    title = {Effective dynamics along given reaction coordinates, and reaction rate theory},
    volume = {195},
    issn = {1359-6640},
    doi = {10.1039/C6FD00147E},
    journal = {Faraday Discussions},
    author = {Zhang, Wei and Hartmann, Carsten and Schütte, Christof},
    year = {2016},
    pages = {365--394},
}

@article{nuske_finite-data_2023,
    title = {Finite-{Data} {Error} {Bounds} for {Koopman}-{Based} {Prediction} and {Control}},
    volume = {33},
    doi = {10.1007/s00332-022-09862-1},
    number = {1},
    journal = {Journal of Nonlinear Science},
    author = {Nüske, Feliks and Peitz, Sebastian and Philipp, Friedrich and Schaller, Manuel and Worthmann, Karl},
    year = {2023},
}

@book{teschl_mathematical_2014,
    address = {Providence, Rhode Island},
    series = {Graduate {Studies} in {Mathematics}},
    title = {Mathematical {Methods} in {Quantum} {Mechanics}},
    volume = {157},
    isbn = {9781470418885 9781470419004 9781470425630 9781470417048},
    url = {https://www.ams.org/gsm/157},
    doi = {10.1090/gsm/157},
    language = {en},
    urldate = {2026-08-25},
    publisher = {American Mathematical Society},
    author = {Teschl, Gerald},
    month = nov,
    year = {2014},
}

@book{bakry_analysis_2014,
    address = {Cham},
    series = {Grundlehren der mathematischen {Wissenschaften}},
    title = {Analysis and {Geometry} of {Markov} {Diffusion} {Operators}},
    volume = {348},
    copyright = {https://www.springernature.com/gp/researchers/text-and-data-mining},
    isbn = {9783319002262 9783319002279},
    url = {https://link.springer.com/10.1007/978-3-319-00227-9},
    doi = {10.1007/978-3-319-00227-9},
    language = {en},
    urldate = {2026-08-25},
    publisher = {Springer International Publishing},
    author = {Bakry, Dominique and Gentil, Ivan and Ledoux, Michel},
    year = {2014},
}

@misc{kingma_adam_2014,
    title = {Adam: {A} {Method} for {Stochastic} {Optimization}},
    copyright = {arXiv.org perpetual, non-exclusive license},
    shorttitle = {Adam},
    url = {https://arxiv.org/abs/1412.6980},
    doi = {10.48550/ARXIV.1412.6980},
    urldate = {2026-08-25},
    publisher = {arXiv},
    author = {Kingma, Diederik P. and Ba, Jimmy},
    year = {2014},
}

@misc{bradbury_jax_2018,
    title = {{JAX}: composable transformations of {Python}+{NumPy} programs},
    url = {http://github.com/jax-ml/jax},
    author = {Bradbury, James and Frostig, Roy and Hawkins, Peter and Johnson, Matthew James and Katariya, Yash and Leary, Chris and Maclaurin, Dougal and Necula, George and Paszke, Adam and VanderPlas, Jake},
    year = {2018},
}

\appendix

\section{Additional Validation Results}
\label{sec:app_validation}
In this section we show additional validation results for the low-dimensional model systems. Figure~\ref{fig:ME_lambda_Bsp} shows the cost-to-go $\lambda$ as a function of data size $m$ and prefactor $\sigma$ if a B-spline basis is used.

\begin{figure}[h!]
    \centering
\includegraphics[width=\linewidth]{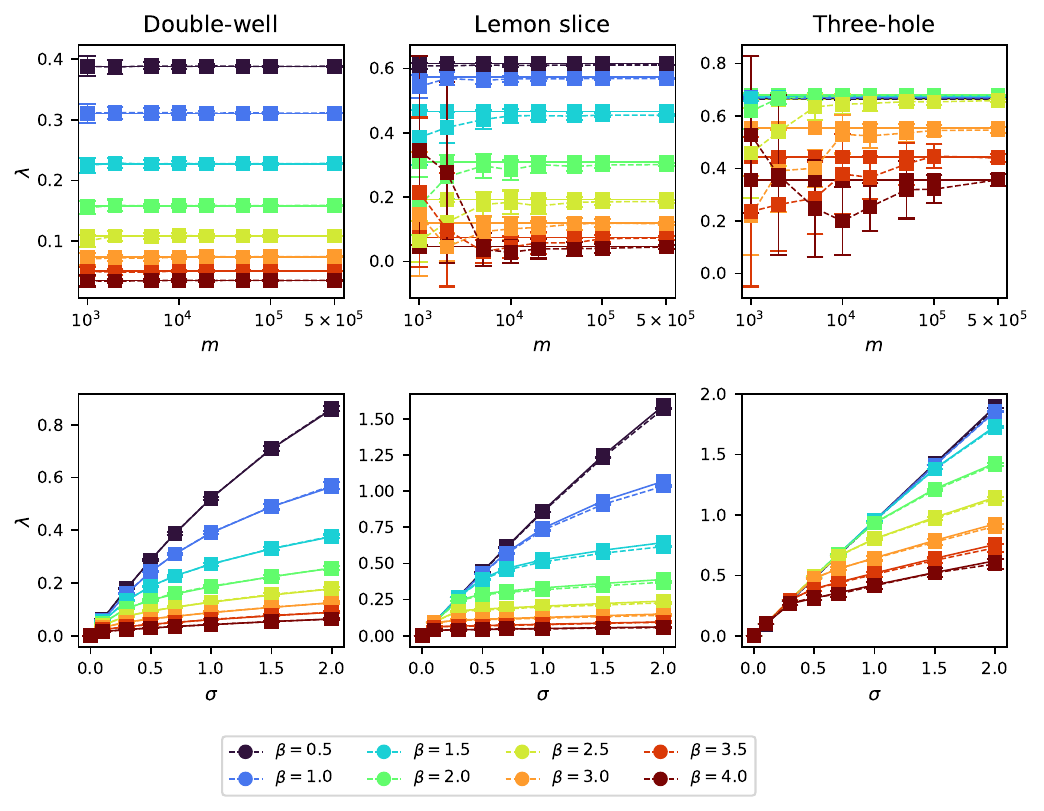}\caption{As Fig.~\ref{fig:ME_lambda_RFF}, using a B-spline basis in place of RFF.
    Left (Double-well): $n=30$ and bottom panel fixes $m=10^{4}$.
    Middle (Lemon slice): $n=15\times15$ and bottom panel fixes $m=5\times10^{5}$.
    Right (Three-hole): $n=10\times 10$ and bottom panel fixes $m=5\times10^{5}$.}
    \label{fig:ME_lambda_Bsp}
\end{figure}

Figure~\ref{fig:ME_dw_bias_RFF} shows the reconstruction of the optimal bias potential for the double-well example and the RFF basis.
\begin{figure}[h!]
    \centering
\includegraphics[width=0.8\linewidth]{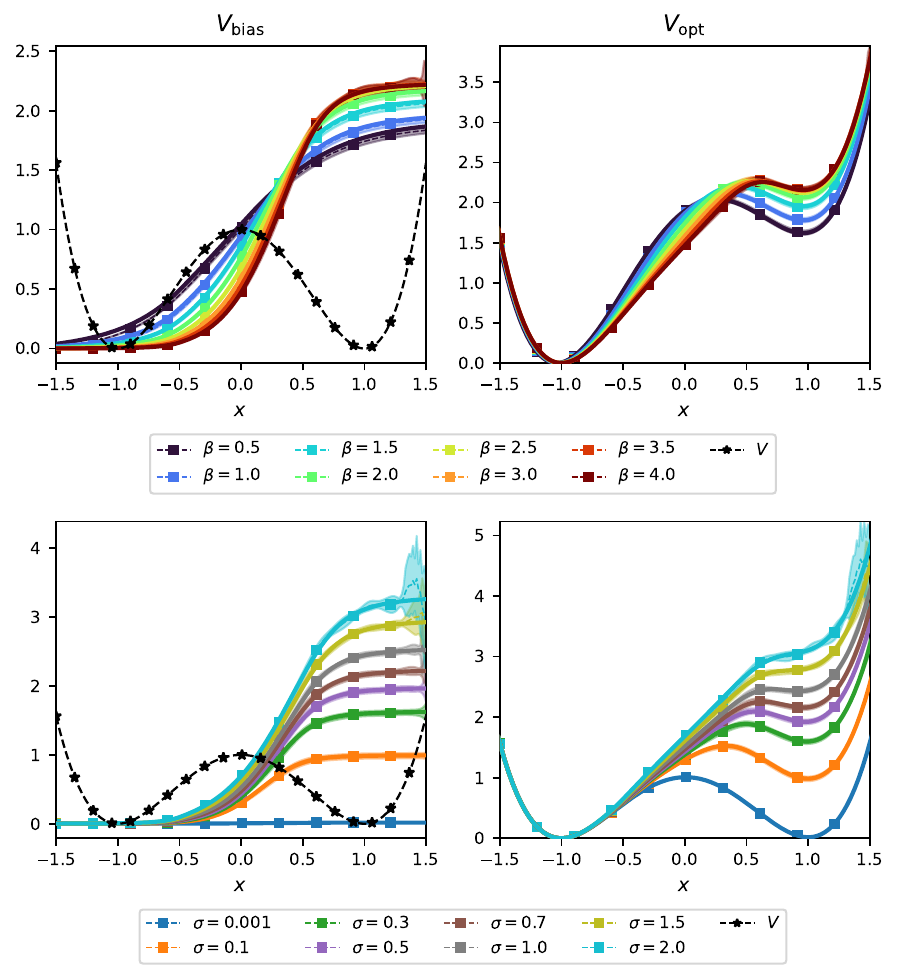}
    \caption{As Fig.~\ref{fig:ME_dw_bias_Bsp} using RFF basis at $m = 10^{4}$ and  $\gamma_{\mathrm{rff}}=0.3$.}
    \label{fig:ME_dw_bias_RFF}
\end{figure}

\section{Proof of Proposition~\ref{prop:ev_problem_gen}}
\label{app:derivations}
We show the steps to derive the eigenvalue equation~\eqref{eq:ev_problem} and the statements of Proposition~\ref{prop:ev_problem_gen}.

\begin{proof}[Proof of Proposition~\ref{prop:ev_problem_gen}]
We proceed in five steps.

\medskip
\noindent\textbf{Step 1 (Existence of the ergodic HJB pair).}
We assume that the running cost $c(\bx,\bu) := \sigma \ell(\bx) + \frac{\beta}{4}|\bu|^2$ and the potential $V$ are such that the ergodic Hamilton-Jacobi-Bellman equation can be solved (see~\cite[Thm.~3.6.6]{arapostathis2012ergodic} for conditions): there exists $W \in C^2(\mathbb{R}^d)$ and a constant $\lambda\in\mathbb{R}$, together satisfying $W(0)=0$, $\inf W > -\infty$, and
\begin{equation}\label{ieq:rho_rhostar}
    \lambda\le\lambda^{*},
\end{equation}
such that the ergodic HJB equation for~\eqref{eq:ergodic_ocp} holds true:
\begin{equation}
\label{eq:hjb_min}
\min_{\bu \in \mathbb{R}^k}\Big[ \cL W(\bx) + \bu\cdot \nabla W(\bx) + \sigma\ell(\bx) + \frac{\beta}{4}\|\bu\|^2 \Big] = \lambda \qquad \forall\, \bx \in \mathbb{R}^d.
\end{equation}

\medskip
\noindent\textbf{Step 2 (Pointwise minimization).}
For fixed $p\in\R^k$, completing the square gives
$$
\bu\cdot p + \frac{\beta}{4}\|\bu\|^2 = \frac{\beta}{4}\Big\|\bu+\frac{2}{\beta}p\Big\|^2 - \frac{1}{\beta}\|p\|^2 \;\ge\;-\frac{1}\beta\|p\|^2,
$$
with equality iff $\bu=-\frac{2}{\beta}p$. Setting $p=\nabla W(\bx)$ in~\eqref{eq:hjb_min},

\begin{equation}
\label{eq:ustar}
\bu^*(\bx) = -\frac{2}{\beta}\nabla W(\bx),
\end{equation}
and~\eqref{eq:hjb_min} reduces to the semi-linear equation
\begin{equation}
\label{eq:semilinear-hjb}
\cL W(\bx) - \frac{1}{\beta}\|\nabla W(\bx)\|^2 + \sigma\ell(\bx) = \lambda \quad \forall\, \bx\in\R^d.
\end{equation}
Defining $V_{\mathrm{bias}}:= \frac{2}{\beta}W$ turns~\eqref{eq:ustar} into
$\bu^*(\bx)=-\nabla V_{\mathrm{bias}}(\bx)$, which is claim~(ii).

\medskip
\noindent\textbf{Step 3 (Logarithmic transform).}
Set $\varphi:= e^{-W} > 0$, so $W=-\log\varphi$ and
$V_{\mathrm{bias}}=-\frac{2}{\beta}\log\varphi$. Then
$$
\nabla W = -\frac{\nabla\varphi}{\varphi}, \qquad
\Delta W = -\frac{\Delta\varphi}{\varphi} + \frac{\|\nabla\varphi\|^2}{\varphi^2},
$$
so, using $\cL=-\nabla V\cdot\nabla+\beta^{-1}\Delta$,
$$
\cL W = \frac{\nabla V\cdot\nabla\varphi}{\varphi} - \frac{1}{\beta}\frac{\Delta\varphi}{\varphi} + \frac{1}{\beta}\frac{\|\nabla\varphi\|^2}{\varphi^2},
\qquad
\frac{1}{\beta}\|\nabla W\|^2 = \frac{1}{\beta}\frac{\|\nabla\varphi\|^2}{\varphi^2}.
$$
The quadratic terms cancel exactly in $\cL W - \frac{1}{\beta}\|\nabla W\|^2$, leaving
$$
\cL W - \frac{1}{\beta}\|\nabla W\|^2
= \frac{\nabla V\cdot\nabla\varphi}{\varphi} - \frac{1}{\beta}\frac{\Delta\varphi}{\varphi}
= -\frac{1}{\varphi}\cL\varphi.
$$
Substituting into~\eqref{eq:semilinear-hjb} and multiplying by $\varphi$,
\begin{equation}
(\sigma\ell - \cL)\varphi = \lambda\,\varphi,
\end{equation}
which is exactly~\eqref{eq:ev_problem}: no rescaling of $\cL$ is needed, because the control-cost coefficient $\beta/4$ in~\eqref{eq:running_cost} was chosen precisely to make this cancellation exact.

\begin{remark}[Why $\beta/4$, precisely]\label{rem:why-beta4}
For a general coefficient $\eta$ in place of $\beta/4$ in~\eqref{eq:running_cost}, the
minimizer becomes $u^*=-\frac{1}{2\eta}\nabla W$ and the semilinear HJB equation reads
$\cL W-\tfrac{1}{4\eta}\|\nabla W\|^2+\sigma\ell=\lambda$. The same substitution
$\varphi=e^{-W}$ then leaves
$$
\cL W - \frac{1}{4\eta}\|\nabla W\|^2
= -\frac{\cL\varphi}{\varphi} + \left(\frac{1}{\beta}-\frac{1}{4\eta}\right)\frac{\|\nabla\varphi\|^2}{\varphi^2},
$$
so the nonlinear term cancels, linearizing the equation into~\eqref{eq:ev_problem}, if and only if $\eta=\beta/4$.

Equivalently: allowing a \emph{tuned} substitution $W:=-c_1\cdot\log\varphi$ with
$c_1=4\eta/\beta$ linearizes the semilinear HJB equation for \emph{any} $\eta>0$,
but leaves the rescaled equation
$\big(\sigma\ell-c_1\cL\big)\varphi=\lambda\varphi$;
only $\eta=\beta/4$ makes $c_1=1$ and removes the rescaling factor
$4\eta/\beta$ from in front of $\cL$, giving~\eqref{eq:ev_problem} exactly,
against the bare generator $\cL$.
\end{remark}

\medskip
\noindent\textbf{Step 4 (Ground state and principal eigenvalue).}
The uncontrolled Langevin dynamics are reversible with respect to $\diff\mu(\bx)\propto e^{-\beta V(\bx)}\,\diff\bx$, and $\cL$ is essentially self-adjoint on $\mathbb{L}^2(\mu)$. Moreover, for sufficiently
regular $f$,
$$
\langle f,-\cL f\rangle_\mu = \frac{1}{\beta}\int_{\mathbb{R}^d} \|\nabla f\|^2\,\diff\mu.
$$
Consequently, the operator $\sigma\ell-\cL$
is essentially self-adjoint and nonnegative, with quadratic form
$$
\langle f,(\sigma\ell-\cL)f\rangle_\mu = \frac{1}{\beta}\int\|\nabla f\|^2\,\diff\mu + \sigma\int\ell f^2\,\diff\mu.
$$  
Its lowest eigenvalue is characterized by the Rayleigh--Ritz formula
$$
\lambda_0 = \inf_{f\neq0} \frac{\beta^{-1}\int\|\nabla f\|^2\,\diff\mu+\sigma\int\ell f^2\,\diff\mu}{\int f^2\,\diff\mu}.
$$
Since the eigenfunction $\varphi=e^{-W}$ constructed in Step~3 is strictly
positive, the principal-eigenvalue characterization identifies
$\lambda=\lambda_0$.

\medskip
\noindent\textbf{Step 5 (Verification and optimality).} Steps 1-4 constructed a pair $(W,\lambda)$ solving the ergodic HJB~\eqref{eq:hjb_min}, with minimizing control $\bu^*(\bx) = -\frac{2}{\beta}\nabla W(\bx)$ and $\lambda$ the smallest eigenvalue of $\sigma\ell-\cL$ on $\mathbb{L}^{2}(\mu)$. By the standard verification theorem for ergodic control (\cite[Thm.~3.5.6/3.5.8 ]{arapostathis2012ergodic}), any such solution pair automatically satisfies $\lambda=\lambda_{\bu^*}=\lambda^*\ge0$ for every $\bx$ -- proving claim (i) -- with the associated feedback control $\bu^{*}$ being stable, proving claim (ii).
\end{proof}

\end{document}